\documentclass[12pt]{amsart}
\usepackage[margin=1in]{geometry}
\usepackage{amsfonts}
\usepackage{amssymb}
\usepackage[dvips]{graphics}
\usepackage{epsfig}
\usepackage{euscript}
\usepackage{color}
\usepackage{tikz}
\usepackage [all,arc,curve,color, arrow, matrix,frame]{xy}
\newcommand*\circled[1]{\tikz[baseline=(char.base)]{
            \node[shape=circle,draw,inner sep=1pt] (char) {#1};}}

 \newtheorem{thm}{Theorem}[section]
 \newtheorem{cor}[thm]{Corollary}
 
 \newtheorem{prop}[thm]{Proposition}
 \theoremstyle{definition}
 
 \theoremstyle{remark}

 \numberwithin{equation}{section}
\numberwithin{figure}{section}

\newcommand{\be}{\begin{equation}}
\newcommand{\ee}{\end{equation}}
\newcommand{\bea}{\begin{eqnarray}}
\newcommand{\eea}{\end{eqnarray}}

\newcommand{\C}{\mathbb{C}}
\newcommand{\R}{\mathbb{R}}

\newcommand{\Z}{\mathbb{Z}}

\newcommand{\N}{\mathbb{N}}

\DeclareMathOperator{\supp}{supp}

\DeclareMathOperator{\dist}{dist}

\def\idty{{\mathchoice {\mathrm{1\mskip-4mu l}} {\mathrm{1\mskip-4mu l}} %
{\mathrm{1\mskip-4.5mu l}} {\mathrm{1\mskip-5mu l}}}}

\numberwithin{equation}{section}

\begin{document}

\title[The XXZ Quantum Spin Chain]{The infinite XXZ quantum spin chain revisited:\\ Structure of low lying spectral bands and gaps}

\author[C. Fischbacher]{Christoph Fischbacher$^1$}
\address{$^1$ School of Mathematics, Statistics and Actuarial Science, University of Kent, Canterbury, Kent CT2 7NF, UK}
\email{cf299@kent.ac.uk}

\author[G. Stolz]{G\"unter Stolz$^2$}
\address{$^2$ Department of Mathematics\\
University of Alabama at Birmingham\\
Birmingham, AL 35294, USA}
\email{stolz@uab.edu}

\date{}

%
\begin{abstract}

We study the structure of the spectrum of the infinite XXZ quantum spin chain, an anisotropic version of the Heisenberg model. The XXZ chain Hamiltonian preserves the number of down spins (or particle number), allowing to represent it  as a direct sum of $N$-particle interacting discrete Schr\"odinger-type operators restricted to the fermionic subspace. In the Ising phase of the model we use this representation to give a detailed determination of the band and gap structure of the spectrum at low energy. In particular, we show that at sufficiently strong anisotropy the so-called droplet bands are separated from higher spectral bands uniformly in the particle number. Our presentation of all necessary background is self-contained and can serve as an introduction to the mathematical theory of the Heisenberg and XXZ quantum spin chains.

\end{abstract}

\maketitle

\section{Introduction} \label{sec:intro}
Quantum spin systems have been widely used in theoretical physics and more recently in information theory as models in the study of many-body quantum systems. Single non-interacting quantum spins are very easy to describe, as they act in low finite-dimensional Hilbert spaces ($\C^2$ in case of a $1/2$-spin). Quantum spin systems thus become models which allow to almost exclusively focus on the complicated effects of interactions. In this spirit, our personal interest in quantum spin systems arose from trying to identify quantum many-body systems in which disorder effects (such as localization) can be rigorously studied. From this point of view the work presented here can be considered as a preliminary investigation of the quantum XXZ chain subject to an exterior magnetic field, which we plan to continue in future work by studying the case of random magnetic field.

The XXZ chain is a one-dimensional anisotropic version of the spin-$1/2$ Heisenberg (or XXX) model. One of the XXZ chain's key properties, making it accessible to rigorous investigation, is that it preserves the {\it particle number}, i.e.\ the number of, say, down spins in a spin configuration. 

One of our goals here is to reiterate that this property can be understood in terms of the fact that the Hamiltonian describing the XXZ spin chain is unitarily equivalent to a direct sum of $N$-body discrete Schr\"odinger-type operators, each restricted to the fermionic subspace of anti-symmetric states and involving attractive nearest neighbor interactions. While many aspects of this have been used at least implicitly in the physics and mathematics literature on the XXX and XXZ chain (in particular, the resulting representation (\ref{eq:equivform}) below is well known), this connection seems to be widely unknown among mathematicians working in the spectral theory of Schr\"odinger operators. Thus we hope that our work can trigger further applications of the wealth of knowledge in Schr\"odinger operator theory to quantum spin systems. It is mostly for this reason that we have written this paper from the pedagogical perspective of not only presenting some new results, but also surveying known facts from a somewhat different point of view and with self-contained proofs, trying to make this knowledge accessible to a new and wider audience.

Much of the mathematical interest in the Heisenberg and XXZ chains has arisen from the fact that their infinite volume versions provide examples of exactly solvable systems. This was accomplished by rigorous implementations of the Bethe ansatz, originating from \cite{Bethe}, leading to an exact diagonalization in terms of a Plancherel formula or complete generalized eigenfunction expansion. For the Heisenberg chain this is done in the series of papers \cite{Thomas1977, BabbittThomas1977, BabbittThomas1978, BabbittThomas1979}, which also discusses consequences for the scattering theory of the Heisenberg chain and that it provides an example of a completely integrable system. Later, the extension of the completeness result to the anisotropic XXZ chain was discussed in \cite{BabbittGutkin1990, Gutkin1993, Gutkin2000}.

As far as we can see, some of the technical details of these results have never been published. But note that \cite{Borodinetal} announces a new proof of completeness of the Bethe ansatz for the XXZ chain as work in preparation. Among recent works on the Bethe ansatz for the Heisenberg chain we also mention \cite{HNS} which provides many additional references.

Our main interest here, motivated by future applications to the XXZ chain in exterior magnetic field, lies in the structure of the low energy bands and gaps of the spectrum of the infinite volume XXZ Hamiltonian and its restrictions to the $N$-particle subspaces. Mathematical results on the ground state properties and low energy spectrum of the XXZ chain can be found in \cite{KN1, KN2, NS, NSS, Starr}, where further references, including to the physics literature, are provided. These works, as well as ours, are generally concerned with the {\it Ising phase} of the XXZ chain, where the $Z$-term (Ising term) dominates the X and Y-terms in the Hamiltonian. 

Our approach in this paper could be characterized as ``semi-hard''. Following the approach in \cite{NSS}, we use a rigorous version of the Bethe ansatz to study the low lying spectrum of the infinite XXZ chain, in particular the so-called {\it droplet bands}. In a sufficiently strong Ising regime this indeed gives the full spectral resolution of the XXZ Hamiltonian at the bottom of its spectrum. Then we use ``softer'' methods, in particular an HVZ-type theorem, to identify higher energy bands without attempting to prove completeness or explicitly describing all generalized eigenfunctions. 

We start in Section~\ref{sec:finitechain} by introducing the {\it finite} XXZ chain and recalling some basic facts, specifically particle number preservation and the structure of ground states.  What we present here can be considered common wisdom, which we learned in private communication from B.\ Nachtergaele. Logically Section~\ref{sec:finitechain} is independent from the rest of our paper, which is concerned with the {\it infinite} XXZ chain, but Section~\ref{sec:finitechain} serves to motivate what follows.

In Section~\ref{sec:infinitechain} we essentially follow \cite{NSS} in the definition of the infinite XXZ chain, also including an exterior magnetic field of variable strength in the $Z$ direction. This can be understood as an explicit construction of the GNS Hilbert space for the infinite volume XXZ chain relative to the all-spins-up vacuum vector and is enabled by particle number preservation (i.e.\ we may first introduce the restriction of the infinite XXZ chain to states with fixed particle number and then take a direct sum).

Then we show in Proposition~\ref{prop:Schrodinger} of Section~\ref{sec:Schrodinger} that the restriction of the infinite XXZ chain to the $N$-particle sector is unitarily equivalent to an $N$-body discrete Schr\"odinger operator restricted to the fermionic subspace. Due to working with the one-dimensional XXZ chain, this operator can also be directly expressed as a Schr\"odinger-type operator over the subgraph of ordered lattice points $x_1 < x_2 < \ldots < x_N$ of $\Z^N$, as pointed out in (\ref{eq:equivform}) below.

Much of the rest of our work, Sections~\ref{sec:droplet} to \ref{sec:gap}, is devoted to a thorough study of the spectrum of the free infinite XXZ chain (i.e.\ in the absence of an exterior field). In this, based on the insights from Section~\ref{sec:Schrodinger}, we are strongly guided by ideas from $N$-body scattering theory.

The construction and exact determination of the {\it droplet spectrum} for each $N$-particle sector of the XXZ chain in Section~\ref{sec:droplet}, while quite closely following \cite{NSS} in the details and using a rigorous version of the Bethe ansatz, thus is seen as determining the spectrum corresponding to one particular scattering channel, namely the one in which all $N$ particles form a single cluster. Then in Section~\ref{sec:HVZ}, via a classical approach from scattering theory, we use an HVZ-type theorem to determine other contributions to the spectrum corresponding to more complicated cluster decompositions of the $N$ particles. We do this for a general class of discrete $N$-body Schr\"odinger operators, only assuming that the interaction potential decays to zero and the kinetic energy operator is translation invariant. Proposition~\ref{prop:HVZ} provides an HVZ-type theorem for the full $N$-particle Hilbert space, which is then shown in Proposition~\ref{prop:HVZsub} to also hold for the restrictions of the $N$-body operators to the symmetric and anti-symmetric subspaces.
 
When applying the results of Section~\ref{sec:HVZ} to our particular model, in light of the attractiveness of the interaction, it is to be expected that the droplet band is the lowest band of the spectrum. As far as the ground state is concerned this is seen already in Section~\ref{sec:droplet}, e.g.\ Corollary~\ref{cor:specmin} (and also contained in \cite{NSS}), which says that the bottom of the droplet spectrum gives the spectral minimum of the XXZ chain in each $N$-particle sector. 

What we consider much more important for potential later applications is Proposition~\ref{prop:existencegap} in Section~\ref{sec:gap}, which can be thought of as the main result of our paper: If the anisotropy parameter $\Delta$ characterizing the Ising phase is sufficiently large ($\Delta>3$ will suffice), then the entire droplet band is strictly separated by a gap from the rest of the spectrum in the $N$-particle sector and this gap can be chosen {\it uniform} in the particle number $N$. In fact, $\Delta>3$ will be needed to cover the somewhat trivial case $N=1$. For all other particle numbers, $\Delta>2$ suffices to create a uniform gap and we will see that these values are sharp.

This also means that for $\Delta>3$ the infinite XXZ chain has a non-trivial spectral gap above the union of all $N$-particle droplet bands. This may be considered an analogue to a result for the {\it finite} XXZ chain proven in \cite{NS} (where, however, the size of $\Delta$ required for this is not explicitly quantified).

In Section~\ref{sec:remarks} we conclude with some remarks on the consequences of exposing the XXZ chain to an exterior magnetic field and, in particular, to a {\it random} field. We mention some initial results on this which were proven in the thesis \cite{F}. The latter also contains most of the other results presented here, often with more detailed proofs. We will also comment on perspectives provided by our results for proving localization properties of the XXZ chain in random field near its spectral minimum.

\section*{Acknowledgements}

G.\ S.\ was supported in part by NSF Grant DMS-1069320. C.\ F.\ was supported by the Max-Weber-Programm Bayern and received travel support from NSF Grant DMS-1069320 as well as from the Master Program in Theoretical and Mathematical Physics as part of the Elitenetwork of Bavaria.

It's a pleasure to acknowledge many fruitful discussions with Bruno Nachtergaele, Robert Sims and Shannon Starr on quantum spin systems, in general, as well as on the XXZ chain, in particular. We also acknowledge hospitality at the Erwin Schr\"odinger International Research Institute for Mathematical Physics in Vienna, where part of this work was done. C.\ F.\ also is grateful for his education at the LMU M\"unchen, in particular for support and encouragement provided by Peter M\"uller who co-advised the thesis \cite{F}, as well as for hospitality provided to him at the UAB Mathematics Department during two extended visits.

\section{The finite XXZ spin chain in transversal field} \label{sec:finitechain}

For much of our work we will consider the {\it infinite} XXZ chain, but it is instructive to start by recalling the {\it finite} XXZ chain and review some of its properties.

The canonical basis vectors in $\C^2$ will be denoted by
\begin{equation}
e_0 = \begin{pmatrix} 1\\0 \end{pmatrix} = |\uparrow\rangle \quad \mbox{and} \quad e_1 = \begin{pmatrix} 0\\1 \end{pmatrix} = |\downarrow\rangle,
\end{equation}
interpreted as ``up-spins'' and ``down-spins''. For $L\in \N$ consider the Hilbert space ${\mathcal H}_{[1,L]} = \bigotimes_{x=1}^L \C^2$ with product basis given by 
\begin{equation} \label{eq:productbasis}
\{e_{i_1} \otimes \ldots \otimes e_{i_L}: \;i_k \in \{0,1\},\, k=1,\ldots,L\}.
\end{equation}
 For these basis vectors we will also use symbols such as $|\uparrow \downarrow\rangle$, $|\uparrow \uparrow \ldots \uparrow \rangle_{[1,L]}$, etc.

Let
\begin{equation}
S^1 = \begin{pmatrix} 0 & 1/2 \\ 1/2 & 0 \end{pmatrix}, \quad S^2 = \begin{pmatrix} 0 & -i/2 \\ i/2 & 0 \end{pmatrix}, \quad S^3 = \begin{pmatrix} 1/2 & 0 \\ 0 & -1/2 \end{pmatrix}
\end{equation}
be the standard spin-$1/2$-matrices and
\begin{equation} 
S^+ = S^1 +iS^2 = \begin{pmatrix} 0 & 1 \\ 0 & 0 \end{pmatrix}, \quad S^- = S^1 -iS^2 = \begin{pmatrix} 0 & 0 \\ 1 & 0 \end{pmatrix}
\end{equation}
be the spin raising and lowering operators. 

The $XXZ$ chain is the self-adjoint Hamiltonian
\begin{equation}
H_{[1,L]} = \sum_{x=1}^{L-1} h_{x,x+1}
\end{equation}
on ${\mathcal H}_{[1,L]}$, where
\begin{eqnarray} \label{eq:frfree}
h_{x,x+1} & = & \gamma(\frac{1}{4}\idty - S_x^3 S_{x+1}^3) - \frac{1}{\Delta}(S_x^1 S_{x+1}^1 + S_x^2 S_{x+1}^2) \\
& = &  -\frac{1}{\Delta} (\vec{S}_x \cdot \vec{S}_{x+1}-\frac{1}{4}\idty) - (\gamma-\frac{1}{\Delta}) (S_x^3 S_{x+1}^3 - \frac{1}{4}\idty) \nonumber \\
& = & -\frac{1}{2\Delta}(t_{x,x+1}-\idty) - \frac{1}{4}(\gamma - \frac{1}{\Delta}) (\pi_{x,x+1}-\idty). \nonumber
\end{eqnarray}
Here, as usual, the notation $A_x$ indicates that a matrix $A$ acts on the $x$-component of the tensor product. The exchange operator $t_{x,x+1}$ acts on a product state by exchanging its components at the sites $x$ and $x+1$, while $\pi_{x,x+1}$ acts on the basis vectors (\ref{eq:productbasis}) as the identity if the sites $x$ and $x+1$ have the same spin and as the negative identity if these sites have opposite spins.

The parameters $\gamma$ and $\Delta$ are assumed non-negative, which is often referred to as the ferromagnetic case.
Special cases  include the Heisenberg chain ($\gamma = 1/\Delta$), the Ising chain ($\Delta = \infty$) and the isotropic XY (or XX) model ($\gamma=0$). 

Here we generally consider the Ising phase of the XXZ chain, i.e.\ we assume $\gamma > 1/\Delta >0$. In fact, in this case nothing will be gained by keeping $\gamma$ flexible, thus from now on we set 
\begin{equation} \label{eq:parameters} 
\gamma = 1 \quad \mbox{and} \quad 1<\Delta<\infty.
\end{equation}

Working with the renormalized form of the XXZ chain, given by the extra $\frac{1}{4}\idty$ terms in the first line of (\ref{eq:frfree}), has the consequence that in the Ising phase the two terms in the last line of (\ref{eq:frfree}) are non-negative (this renormalization is also needed to extend the definition of the XXZ chain to infinite volume in the next section). From this one can see that the ground state energy of $H_{[1,L]}$ is $0$ and that the corresponding eigenspace is two-dimensional, spanned by $|\uparrow \uparrow \ldots \uparrow\rangle_{[1,L]}$ and $|\downarrow \downarrow \ldots \downarrow \rangle_{[1,L]}$, the all-spins-up and all-spins-down states (justifying the use of the term ``ferromagnetic'', as it is energetically optimal for all spins to be parallel). This is well known and sometimes referred to by saying that the XXZ chain in the Ising phase is ``frustration free'', allowing to find ground states as the intersection of null spaces of non-negative terms. Details on this are provided in \cite{F}.

Let $\nu = (\nu_x)_{x=1}^L$ such that all $\nu_x\ge 0$ and define 
\begin{equation} \label{eq:finchainfield}
H_{[1,L]}(\nu) = H_{[1,L]} + \sum_{x=1}^L \nu_x {\mathcal N}_x,
\end{equation}
where 
\begin{equation}
\mathcal{N}_x:= S_x^- S_x^+ =\begin{pmatrix} 0 & 0 \\ 0 & 1 \end{pmatrix}_x = \frac{1}{2} \idty - S_x^3.
\end{equation}
Thus, up to an energy shift $H_{[1,L]}(\nu)$ is the Hamiltonian of the XXZ chain in an external transversal magnetic field of variable strength $\nu$. Choosing the field $\nu_x$ non-negative makes $|\uparrow \ldots \uparrow \rangle$ a ground state of $H_{[1,L]}(\nu)$ to energy $E_0=0$. If $\nu_x >0$ for at least one $x$, then $|\uparrow \ldots \uparrow \rangle$ is the unique ground state (as the energy of $|\downarrow \ldots \downarrow \rangle$ is $\sum_x \nu_x$).

Let
\begin{equation}
S_{[1,L]}^3 := \sum_{x=1}^L S_x^3
\end{equation}
be the total magnetization operator. Then
\begin{equation} \label{totalspincom}
[H_{[1,L]}(\nu), S_{[1,L]}^3]=0.
\end{equation}
$S_{[1,L]}^3$ has eigenvalues $\{\frac{L}{2}-N: N=0,1,\ldots,L\}$, with eigenspace corresponding to $\frac{L}{2}-N$ given by
\begin{equation}
{\mathcal H}_{[1,L]}^N := \mbox{span} \left\{ \prod_{j=1}^N S_{x_j}^- |\uparrow \ldots \uparrow \rangle_{[1,L]}: 1\le x_1 <x_2 <\ldots < x_N \le L \right\},
\end{equation}
i.e.\ the subspace of ${\mathcal H}_{[1,L]}$ spanned by all the states with $N$ down spins, called the $N$-magnon space or $N$-magnon sector. By (\ref{totalspincom}), ${\mathcal H}_{[1,L]}^N$ is invariant under $H_{[1,L]}(\nu)$ and we define
\begin{equation}
H_{[1,L]}^N(\nu) = H_{[1,L]}(\nu)\upharpoonright_{{\mathcal H}_{[1,L]}^N}.
\end{equation}

Note that we may also think of ${\mathcal H}_{[1,L]}^N$ as the $N$-particle sector of ${\mathcal H}_{[1,L]}$, where $\Omega = |\uparrow\ldots \uparrow \rangle$ is the vacuum vector and the lowering operators $S_x^-$ take the role of creation operators. Observe that (\ref{totalspincom}) is equivalent to $[H_{[1,L]}(\nu), {\mathcal N}]=0$, where 
\begin{equation}
{\mathcal N} := \sum_{x=1}^L {\mathcal N}_x = \frac{L}{2} \idty - S_{[1,L]}^3
\end{equation}
is the total particle number operator, characterized by ${\mathcal N} f = Nf$ for $f \in {\mathcal H}_{[1,L]}^N$. In this interpretation the ${\mathcal N}_x$ take the role of local particle number operators.

Thus $H_{[1,L]}^N(\nu)$ is the restriction of the XXZ-Hamiltonian to the $N$-particle states.

\section{The infinite XXZ chain} \label{sec:infinitechain}

Here we essentially follow \cite{NSS} in the introduction of the infinite XXZ chain. This can be understood as a direct definition of the GNS representation of the infinite volume limit of the finite XXZ chain relative to the all-spins-up ground state (or vacuum vector).

We start with the countable set of all {\it finite} subsets of $\Z$,
\begin{equation*}
{M}:=\{J\subset\Z : J\; \mbox{finite}\}
\end{equation*}
and let $\mathcal{H}_\Z=\ell^2(M)$ be the space of all square-summable functions from $M$ to $\C$,
\begin{equation*}
\ell^2(M):=\left\{f: \; M\to\C\;\;\mbox{with}\; \|f\|_{\ell^2(M)}:=\Big(\sum_{J\in M}|f(J)|^2\Big)^{1/2}<\infty\right\}\:.
\end{equation*}
We denote its canonical basis by $\{\phi_J\}_{J\in M}$, i.e.\ $\phi_J:M\to\C$ is given by $\phi_{J}(J')=\delta_{JJ'}$. We also have $\langle\phi_J,\phi_{J'}\rangle=\delta_{JJ'}$.

Represent these basis vectors as infinite arrays
\begin{equation} \label{eq:infarrays}
\phi_J = |\cdots\uparrow\cdots\downarrow\cdots\rangle
\end{equation}
with``$\downarrow$'' in the finitely many positions $x\in J$ and ``$\uparrow$'' otherwise. In this representation the all-spins-up vector
\begin{equation}
\phi_{\emptyset} = | \ldots \uparrow \uparrow \uparrow \ldots \rangle
\end{equation}
takes the role of the vacuum vector from which all basis vectors $\phi_J$ can be generated by successive application of the spin lowering operators,
\begin{equation}
\phi_J = \prod_{x\in J} S_x^- | \ldots \uparrow \uparrow \uparrow \ldots \rangle.
\end{equation}

By the rules described in Section~\ref{sec:finitechain} we let the operators $t_{x,x+1}$, $\pi_{x,x+1}$ and thus, via the last identity in (\ref{eq:frfree}), $h_{x,x+1}$ act on the basis vectors $\phi_J$ in the representation (\ref{eq:infarrays}).  

In the following, we denote by $\ell^2_0(M)$ the {\it finite} linear combinations of elements of $\{\phi_J: J\in M\}$ or, in other words, the subspace of finitely supported vectors in $\ell^2(M)$. Note that 

(i) $h_{x,x+1} \phi_J \in \ell_0^2 (M)$ for all $x\in \Z$, $J\in M$, and

(ii) $h_{x,x+1} \phi_J = 0$ if $\{x,x+1\} \cap J = \emptyset$. 

In particular, for any given $J\in M$ we have $h_{x,x+1} \phi_J =0$ for all but finitely many $x$.

This shows that we may define a linear operator $H_{\Z,0}$ on $\ell^2_0(M)$ by 
\begin{equation}
H_{\Z,0} f = \sum_{x\in\Z}h_{x,x+1}f \in \ell^2_0(M).
\end{equation}
By checking that the operators $t_{x,x+1}$ and $\pi_{x,x+1}$ are hermitean and isometric on $\ell_0^2(M)$ it can be verified that $H_{\Z,0}$ is hermitean and non-negative. As $\ell^2_0(M)$ is dense in $\mathcal{H}_\Z = \ell^2(M)$, $H_{\Z,0}$ is symmetric in $\mathcal{H}_\Z$.

For $N=0,1,2,\dots$ consider the subspaces
\begin{equation}
\mathcal{H}_{\Z}^N :=\overline{\mbox{span}\{\phi_J: J\in M_N\}}
\end{equation}
of $\mathcal{H}_\Z$, where $M_N:=\{J\in M\,:\,\#(J)=N\}$. Intuitively, it is clear that $\mathcal{H}_{\Z}^N$ corresponds to the $N$-particle spaces for the finite case, but, with the exception of the one-dimensional vacuum space $\mathcal{H}_{\Z}^0$, all of them are infinite-dimensional. We have
\begin{equation}
\mathcal{H}_\Z=\bigoplus_{N\geq0}\mathcal{H}_{\Z}^N\:.
\end{equation}

Also, the sets $\ell^2_0(M_N) = \mbox{span}\{\phi_J: J\in M_N\}$ are invariant under $H_{\Z,0}$, which is seen in the same way as for finite systems. Thus the restrictions
\begin{equation} 
H_{\Z,0}^N := H_{\Z,0} \upharpoonright_{\ell^2_0(M_N)} 
\end{equation}
are non-negative symmetric operators in $\mathcal{H}_{\Z}^N$.

\begin{prop} \label{prop:formbound}
For every $N=0,1,2,\ldots$, and every $f\in \ell^2_0(M_N)$, we have
\begin{equation} \label{formbound} 
\langle f, H_{\Z,0}^N f \rangle \le N(1+\frac{1}{\Delta})\|f\|^2.
\end{equation}
\end{prop}

\begin{proof}
It is clear how the local particle number operators ${\mathcal N}_x$
act on the basis vectors $\phi_J$ represented by (\ref{eq:infarrays}) and thus on $f\in \ell^2_0(M)$, where ${\mathcal N}_x f =0$ for all but finitely many $x$. Thus we can also let the total particle number operator $\mathcal{N}=\sum_{x\in\Z}\mathcal{N}_x$ act on $\ell^2_0(M)$ and get $\mathcal{N} f = N f$ for all $f\in \ell_0^2(M_N)$.

Next we claim that for any $x\in\Z$ and for any $f\in\ell^2_0(M)$ it holds that
\begin{equation} 
\langle f,h_{x,x+1}f\rangle\leq\frac{1}{2}\left(1+\frac{1}{\Delta}\right)\langle f,\left(\mathcal{N}_x+\mathcal{N}_{x+1}\right)f\rangle\:.
\label{eq:opbound}
\end{equation}
For this purpose, it is justified to consider just the actions of these operators on $\C^2\otimes\C^2$, representing the sites $x$ and $x+1$, as they leave the other entries unaffected and we only sum over finitely many basis elements. Thus, let us look at the $\C^2\otimes\C^2$ eigenbasis 
\begin{equation}
|\uparrow \uparrow \rangle, \quad \frac{|\uparrow\downarrow\rangle+|\downarrow\uparrow\rangle}{\sqrt{2}}, \quad \frac{|\uparrow\downarrow\rangle-|\downarrow\uparrow\rangle}{\sqrt{2}}, \quad |\downarrow \downarrow \rangle
\end{equation}
of $h_{x,x+1}$, which is an eigenbasis of $\mathcal{N}_x+\mathcal{N}_{x+1}$ as well. Thus in this basis (\ref{eq:opbound}) becomes a claim for diagonal matrices, which one can check explicitly.

Now, for $f\in \ell^2_0(M_N)$ we may sum \eqref{eq:opbound} over $x$ (with only finitely many non-trivial terms) and get
\begin{align}
\langle f,H_{\Z,0}^Nf\rangle&=\left\langle f,\sum_x h_{x,x+1}f\right\rangle \leq \frac{1}{2}\left(1+\frac{1}{\Delta}\right) \left\langle f, \sum_x (\mathcal{N}_x+\mathcal{N}_{x+1})f \right\rangle=\\
&=\left(1+\frac{1}{\Delta}\right)\langle f,\mathcal{N}f\rangle = N\left(1+\frac{1}{\Delta}\right)\|f\|^2\:,
\end{align}
which completes the proof.
\end{proof}
\vspace{.5cm}

As the $H_{\Z,0}^N$ are symmetric (in particular densely defined) operators in $\mathcal{H}_{\Z}^N$, we can now conclude from Theorem 4.4(b) in \cite{Weidmann} that the $H_{\Z,0}^N$ are bounded with $\|H_{\Z,0}^N\| \le N(1+\frac{1}{\Delta})$ and thus can be uniquely extended by closure to bounded symmetric and non-negative operators defined on all of $\mathcal{H}_{\Z}^N$. We write $H_{\Z}^N := \overline{H_{\Z,0}^N}$.

It now follows from general facts on direct sums (see e.g.\ Exercise 5.43 in \cite{Weidmann}) that $H_{\Z,0}$ is essentially self-adjoint in $\mathcal{H}_{\Z}$ and that its self-adjoint closure $H_{\Z} := \overline{H_{\Z,0}}$ is non-negative and satisfies
\begin{equation} \label{eq:infvolXXZ}
H_{\Z} = \bigoplus_{N\ge 0} H_{\Z}^N\:.
\end{equation}

It is easy to incorporate a bounded exterior field $\nu = (\nu_x)_{x\in\Z}$, $0\le \nu_x \le \nu_{max}$ for all $x\in \Z$, in the above construction of the infinite XXZ chain: On $\ell_0^2(M)$ let
\begin{equation}
H_{\Z,0}(\nu) f = \sum_{x\in \Z} (h_{x,x+1} + \nu_x {\mathcal N}_x)f,
\end{equation}
which, as before, leaves the subspaces ${\mathcal H}_{\Z}^N$ invariant. By modifying the proof of Proposition~\ref{prop:formbound} one sees that the restrictions $H_{\Z,0}^N(\nu)$ of $H_{\Z,0}(\nu)$ to $\ell_0^2(M_N)$ satisfy
\begin{equation}
\langle f, H_{\Z,0}^N(\nu) f \rangle \le N(1+\frac{1}{\Delta} + \nu_{max}) \|f\|^2.
\end{equation}
Thus the $H_{\Z,0}(\nu)$ have bounded self-adjoint closures $H_{\Z}^N(\nu)$ in $\ell^2(M_N)$ and
\begin{equation} \label{eq:infvolXXZmag}
H_{\Z}(\nu) := \overline{H_{\Z,0}(\nu)} = \bigoplus_{N\ge 0} H_{\Z}^N(\nu).
\end{equation}

\section{The relation of the XXZ chain to discrete Schr\"odinger operators} \label{sec:Schrodinger}

We will now show that $H_{\Z}^N(\nu)$, the restriction of the XXZ Hamiltonian to the $N$-particle subspace, is unitarily equivalent to an $N$-body discrete Schr\"odinger operator restricted to the antisymmetric subspace $\ell^2_a(\Z^N) =\bigwedge^N \ell^2(\Z)$ of $\ell^2(\Z^N) = \bigotimes^N \ell^2(\Z)$, i.e.\ the set of those $f\in \ell^2(\Z^N)$ such that
\begin{equation}
f(x_{\sigma(1)}, \ldots, x_{\sigma(N)}) = \mbox{\rm sgn}(\sigma) f(x_1,\ldots,x_N)
\end{equation}
for all permutations $\sigma \in S_N$. In fact, the equivalence will arise by simply identifying the basis vectors $\phi_{\{x_1,x_2,\ldots,x_N\}}$, $x_1 < \ldots x_N$, of ${\mathcal H}_{\Z}^N$ with the canonical basis vectors of $\ell^2_a(\Z^N)$ given by
\begin{equation} \label{wedgebasis}
\widetilde{e}_{(x_1,x_2,\dots,x_N)}:=e_{x_1}\wedge e_{x_2}\wedge\cdots\wedge e_{x_N}=\frac{1}{\sqrt{N!}}\sum_{\sigma\in S_N} \mbox{\rm sgn}(\sigma) e_{x_{\sigma(1)}}\otimes e_{x_{\sigma(2)}}\otimes\cdots\otimes e_{x_{\sigma(N)}}.
\end{equation}
Here $e_x$ are the usual canonical basis vectors of $\ell^2(\mathcal{\Z})$. To accomplish this more formally, let $U: {\mathcal H}_{\Z}^N \to \ell^2_a(\Z^N)$ be the unique unitary operator with $U\phi_{\{x_1,x_2,\ldots,x_N\}}=\widetilde{e}_{(x_1,x_2,\dots,x_N)}$ for all $x_1< \ldots < x_N$.

Let $h_{0,\Z}$ be the discrete Laplace operator on $\ell^2(\Z)$,
\begin{equation} \label{eq:discLap}
(h_{0,\Z}u)(x)=u(x-1)+u(x+1) \quad \mbox{for all $x\in\Z$},
\end{equation}
and $h_{0,\Z}(\nu) = h_{0,\Z} + \nu$ be the discrete one-particle Schr\"odinger operator with bounded potential $\nu = (\nu_x)_{x\in\Z}$. For any $N \in \N$ let
\begin{equation}
h_{\Z}^N(\nu):=-\frac{1}{2\Delta}\sum_{j=1}^N(h_{0,\Z}(\nu))_j+W
\end{equation}
on $\bigotimes_{j=1}^N\ell^2(\Z)=\ell^2(\Z^N)$. Here $(h_{0,\Z}(\nu))_j$ acts non-trivially as $h_{0,\Z}(\nu)$ on the $j$-th component of simple tensors and $W$ is the multiplication operator by
\begin{equation} \label{eq:W}
W(x_1,\ldots,x_N) = -\#(\{j:x_{j+1}=x_j+1\}).
\end{equation}
Notice that the latter can be written as
\begin{equation}
W(x_1,\dots,x_N)=\sum_{1\leq k<l\leq N} Q(|x_k-x_l|),
\end{equation}
with $Q$ given by
\begin{equation} \label{eq:potential}
Q(r)=\begin{cases}-1,\quad &\text{\emph{if}\:\:}r=1\\\:\:\,\,0,&\text{\emph{if}\:\:}r\neq 1\:.\end{cases}
\end{equation}
Thus $h_{\Z}^N(\nu)$ is a discrete $N$-body Schr\"odinger operator in exterior potential $\nu$ and next-neighbor attractive pair interaction $Q$.

Finally, let $h_{\Z,a}^N(\nu)$ be the restriction of $h_{\Z}^N(\nu)$ to the antisymmetric subspace $\ell^2_a(\Z^N)$.

Up to an energy-shift by the particle number $N$, the operators $H_{\Z}^N(\nu)$ and $h_{\Z,a}^N(\nu)$ are unitarily equivalent:

\begin{prop} \label{prop:Schrodinger}
For all $N=0,1,2,\dots$, we have
\begin{equation} \label{eq:Schrodinger}
H_{\Z}^N(\nu)=U^{-1}(h_{\Z,a}^N(\nu)+N\cdot\idty)U\:.
\end{equation}
\end{prop}

\begin{proof}
Due to the boundedness of the operators on both sides of (\ref{eq:Schrodinger}), we only need to check that this identity holds when applied to the canonical basis vectors. This is how all the following calculations are to be understood.

At first, we need the following two identities of operators:
\begin{equation}
\frac{1}{4}\idty-S_x^3S_{x+1}^3 = -\mathcal{N}_x\mathcal{N}_{x+1}+\frac{1}{2}\mathcal{N}_x+\frac{1}{2}\mathcal{N}_{x+1},
\end{equation}
\begin{equation}
S^1_xS^1_{x+1}+S^2_xS^2_{x+1} = \frac{1}{2}\left(S_x^-S_{x+1}^++S_{x+1}^-S_x^+\right),
\end{equation}
which can be checked by direct inspection. The two-site Hamiltonian $h_{x,x+1}$ can therefore be written as
\begin{equation}
h_{x,x+1}=-\mathcal{N}_x\mathcal{N}_{x+1}+\frac{1}{2}\mathcal{N}_x+\frac{1}{2}\mathcal{N}_{x+1}-\frac{1}{2\Delta}\left(S_x^-S_{x+1}^++S_{x+1}^-S_x^+\right).
\end{equation}
Summing over $x$ and including the magnetic field yields
\begin{align}
H_{\Z}^N(\nu) & =\sum_{x\in\Z}\left[\frac{1}{2}\mathcal{N}_x+\frac{1}{2}\mathcal{N}_{x+1}-\mathcal{N}_x\mathcal{N}_{x+1}-\frac{1}{2\Delta}\left(S_x^-S_{x+1}^++S_{x+1}^-S_x^+\right) + \nu_x {\mathcal N}_x\right]\\&=\mathcal{N}-\sum_{x\in\Z}\left[\mathcal{N}_{x}\mathcal{N}_{x+1} - \nu_x {\mathcal N}_x+\frac{1}{2\Delta}(S_x^-S_{x+1}^++S_{x+1}^-S_x^+)\right]. \nonumber
\end{align}
Now, let $J\in M_k$ and observe that $\mathcal{N}_x\mathcal{N}_{x+1}\phi_J=\phi_J$ if $\{x,x+1\} \subset J$, while it vanishes otherwise. Carrying out the summation over $x$ and writing $J=\{x_1,x_2,\dots,x_N\}$ therefore yields
\begin{equation}
\sum_{x\in\Z}\mathcal{N}_x\mathcal{N}_{x+1}\phi_{\{x_1,x_2,\dots,x_N\}} = \#(\{j:x_{j+1}=x_j+1\})\phi_{\{x_1,x_2,\dots,x_N\}}.
\end{equation}

Combining this with 
\begin{equation}
\left(\mathcal{N} + \sum_x \nu_x \mathcal{N}_x\right) \phi_{\{x_1,\ldots,x_N\}} = \left(N+\sum_{j=1}^N \nu_{x_j} \right) \phi_{\{x_1,\ldots,x_N\}}
\end{equation}
yields that, as operators on ${\mathcal H}_{\Z}^N$,
\begin{equation} \label{eq:proofpart1}
\mathcal{N}-\sum_{x\in\Z}(\mathcal{N}_x\mathcal{N}_{x+1} - \nu_x {\mathcal N}_x) =U^{-1}\left(V+W+N\cdot\idty \right)U,
\end{equation}
where 
\begin{equation} \label{eq:V}
V(x_1,\ldots,x_N) = \nu_{x_1} + \ldots + \nu_{x_N}.
\end{equation}

Moreover, it holds that
\begin{align}
(&S_x^-S_{x+1}^++S_{x+1}^-S_x^+)\phi_{\{x_1,\dots,x_N\}}\\=&\begin{cases}0\qquad&\mbox{if} \;x=x_j, x+1=x_{j+1}\;\mbox{for some} \;j\in\{1,\dots,N\},\\0&\mbox{if}\;\{x,x+1\}\cap\{x_1,\dots,x_N\}=\emptyset,\\\phi_{\{x_1,\dots,x_{j-1},x_j+1,x_{j+1},\dots,x_N\}}&\mbox{if}\;x=x_j,\;x_{j+1}>x+1\;\mbox{for some}\;j\in\{1,\dots,N\},\\\phi_{\{x_1,\dots,x_{j-1},x_j-1,x_{j+1},\dots,x_N\}}&\mbox{if}\;x+1=x_j,\;x_{j-1}<x\;\mbox{for some}\;j\in\{1,\dots,N\}\:. \nonumber
\end{cases}
\end{align}
Summing this expression over $x$ therefore yields
\begin{equation}
\sum_{x\in\Z}(S_x^-S_{x+1}^++S_{x+1}^-S_x^+)\phi_{\{x_1,\dots,x_N\}} = \sum_{j=1}^N \left(\sum_{y\in\{x_j-1,x_j+1\}\setminus\{x_1,\dots,x_N\}}\phi_{\{x_1,\dots,x_{j-1},y,x_{j+1},\dots,x_N\}}\right).
\end{equation}

To finish the proof, it remains to show that
\begin{equation} \label{eq:laplace}
\left(\sum_{j=1}^N(h_{0,\Z})_j\right)\widetilde{e}_{(x_1,\dots,x_N)} =\sum_{j=1}^N\left(\sum_{y\in\{x_j-1,x_j+1\}\setminus\{x_1,\dots,x_N\}}\widetilde{e}_{(x_1,\dots,x_{j-1},y,x_{j+1},\dots,x_N)}\right).
\end{equation}
For this purpose, let $P_a$ be the orthogonal projection onto $\ell^2_a(\Z^N)$, which is explicitly given by $P_a f = \sum_{x_1<\dots<x_N}\langle\widetilde{e}_{(x_1,\dots,x_N)},f\rangle\widetilde{e}_{(x_1,\dots,x_N)}$. For $e_{(x_1,\dots,x_N)}:=e_{x_1}\otimes \ldots \otimes e_{x_N}$ one checks that $P_a e_{(x_1,\dots,x_N)}=\frac{1}{\sqrt{N!}}\widetilde{e}_{(x_1,\dots,x_N)}$ if $x_1, \ldots, x_N$ are pairwise disjoint and $0$ otherwise.

Moreover, the $N$-particle operator $\sum_{j=1}^N(h_{0,\Z})_j$ preserves anti-symmetry, i.e.\ it commutes with the projection $P_a$. 

Thus we can calculate
\begin{eqnarray} \label{eq:proofanti}
\lefteqn{\sum_{j=1}^N(h_{0,\Z})_j \widetilde{e}_{(x_1,\dots,x_N)}} \\ & = & \sqrt{N!} \sum_{j=1}^N(h_{0,\Z})_j P_a e_{(x_1,\ldots,x_N)} = \sqrt{N!} P_a \sum_{j=1}^N(h_{0,\Z})_j e_{(x_1,\ldots,x_N)} \nonumber \\
& = & \sqrt{N!} P_a \sum_{j=1}^N (e_{(x_1,\ldots,x_{j-1},x_j+1,x_{j+1},\ldots,x_N)} + e_{(x_1,\ldots, x_{j-1}, x_j-1,x_{j+1},\ldots, x_N)}) \nonumber \\
& = & \sum_{j=1}^N\left(\sum_{y\in\{x_j-1,x_j+1\}\setminus\{x_1,\dots,x_N\}}\widetilde{e}_{(x_1,\dots,x_{j-1},y,x_{j+1},\dots,x_N)}\right), \nonumber
\end{eqnarray}
which is the desired result.
\end{proof}

The identity shown in (\ref{eq:proofanti}) can be used further to identify the restriction of $\sum_{j=1}^N(h_{0,\Z})_j$ to the anti-symmetric subspace with the adjacency operator on a suitable graph. Namely, let
\begin{equation}
\mathcal{X}^N:=\{(x_1,x_2,\dots,x_N)\in\Z^N\::\: x_1<x_2<\cdots<x_N\}
\end{equation}
be the set of ordered integer $N$-tuples, which we think of as a subgraph of $\Z^N$, i.e.\ edges are given by next neighbors with respect to the $\ell^1$-distance.

Moreover, let $h^{(\mathcal{X}^N)}_0$ denote the adjacency operator on $\mathcal{X}^N$, i.e., for any $f\in\ell^2(\mathcal{X}^N)$ and for any $x\in\mathcal{X}^N$,
\begin{equation}
\left(h^{(\mathcal{X}^N)}_0f\right)(x)=\sum_{y\in \mathcal{X}^N:\|x-y\|_1=1}f(y)\:.
\end{equation}

Next neighbors of a point $(x_1,\ldots,x_N) \in \mathcal{X}^N$ are exactly the points
\begin{equation}
(x_1,\dots,x_{j-1},y,x_{j+1},\dots,x_N), \quad j\in \{1,\ldots,N\}, \quad y\in\{x_j-1,x_j+1\}\setminus\{x_1,\dots,x_N\}.
\end{equation}
Thus, by (\ref{eq:proofanti}), under the identification of the basis vectors $\widetilde{e}_{(x_1,\dots,x_N)}$ and $e_{(x_1,\ldots,x_N)}$, the restriction of $\sum_{j=1}^N(h_{0,\Z})_j$ to $\ell^2_a(\Z^N)$ is unitarily equivalent to $h^{(\mathcal{X}^N)}_0$.

Including also the terms in (\ref{eq:proofpart1}) we thus have
\begin{equation} \label{eq:equivform}
H_{\Z}^N(\nu) = - \frac{1}{2\Delta} h^{(\mathcal{X}^N)}_0 + V + W + N\cdot \idty,
\end{equation}
as operators on $\ell^2(\mathcal{X}^N)$, where $W$ and $V$ are given by (\ref{eq:W}) and (\ref{eq:V}). Here we slightly abuse notation by not renaming a unitarily equivalent operator (which arises by simply identifying canonical basis vectors). For the rest of this paper the notation $H_{\Z}^N(\nu)$ will refer to (\ref{eq:equivform}).

In summary, we have shown that the Hamiltonian of the infinite XXZ chain in exterior magnetic field is unitarily equivalent to an infinite direct sum (over $N$) of discrete interacting $N$-body Schr\"odinger operators or, equivalently, of the Schr\"odinger-type operators (\ref{eq:equivform}). 

This is easily seen to extend to the finite XXZ chain (\ref{eq:finchainfield}), where instead of the discrete Laplacian (\ref{eq:discLap}) on $\Z$ one uses its truncation to the discrete interval $[1,L]$. In this case the direct sum (\ref{eq:infvolXXZmag}) is finite, ranging from $N=0$ to $N=L$, and all spaces are finite-dimensional.

\section{The droplet spectrum} \label{sec:droplet}

The purpose of this section is the computation of the lowest energy band of the operators $H_{\Z}^N = H_{\Z}^N(0)$, given by (\ref{eq:equivform}) with $V=0$ and corresponding to the restriction to the $N$-particle subspace of the infinite XXZ chain without exterior field. For reasons that will become clearer in the following this band will be called the doplet band. 

Much of the content of this section is due to \cite{NSS}. In fact, what we do here can be considered as a re-organized presentation of material from \cite{NSS}. Instead of reasoning within the framework of spin systems we take a more ``Schr\"odinger operator point of view'', made possible by the connections pointed out in Section~\ref{sec:Schrodinger} above.

One important observation is the fact that in the absence of an exterior field our system shows translation invariance. Thinking of the $N$ particles (or down-spins) as located along the $\Z$-axis and interacting with each other, we understand that it is only their \emph{relative} distance that matters, which allows us to reduce the problem by one dimension. 

As a result, we are going to obtain low energy states that -- because of the attractive interaction potential $W$ -- are localized at small particle distances, i.e.\ they form a droplet.

\subsection{A coordinate change} \label{chap:coord}

For $N\ge 2$ (with $N=1$ being trivial) the translation invariance of our system becomes clearer after the unitary coordinate change $V_0: \ell^2(\mathcal{X}^N) \to \ell^2(\Z\times\N^{N-1})$ given by
\begin{equation} 
(V_0f)(x,n_1,\dots,n_{N-1}) = f(x, x+n_1, x+n_1+n_2,\dots,x+n_1+n_2+\dots+n_{N-1}).
\end{equation}
Note that $(V_0^{-1}f)(x_1,\dots,x_N)=f(x_1,x_2-x_1,\dots,x_{N}-x_{N-1})$ and that as new variables $(x,n_1,\ldots,n_{N-1})$ we use the position $x$ of the first particle and the distances $n_i$ between the $i^{th}$ and $(i+1)^{th}$ particle. 

To simplify notation, set $m:=N-1$. Let us consider the unitarily equivalent operator
\begin{equation}
\widetilde{H}^N := V_0 H_{\Z}^N V_0^{-1} = -\frac{1}{2\Delta} h_0^{(\Z\times\N^m,\Gamma)} + \widetilde{W} + N\cdot \idty,
\end{equation}
where $h_0^{(\Z\times\N^m,\Gamma)} := V_0 h_0^{(\mathcal{X}^N)} V_0^{-1}$ and  $\widetilde{W}:=V_0WV_0^{-1}$. 

Here $h_0^{(\Z\times\N^m,\Gamma)}$ is the adjacency operator for the graph with vertices $\Z \times \N^m$ and edges $\Gamma$. However, the edges are not $\ell^1$-next neighbors, but images of pairs of next neighbors on $\mathcal{X}^N$ under $V_0$. This means that $(x,n_1,\ldots,n_m)$ and $(x',n_1',\ldots,n_m')$ form an edge in $\Gamma$ if and only if $(x,x+n_1,\ldots,x+n_1+\ldots+n_m)$ and $(x',x'+n_1',\ldots,x'+n_1'+\ldots+n_m')$ are $\ell^1$-next neighbors.

Let us illustrate the effect of the change of variables $V_0$ for the simplest case $N=2$. Figure \ref{fig:1} visualizes the action of $H_{\Z}^2$. Each edge is weighted with $-\frac{1}{2\Delta}$ and the vertices are labeled by the value of the potential $W+ 2\cdot \idty$.

\begin{figure}[htbp]
\centering
 \fbox{  \xygraph{ 
	!{<0cm,0cm>;<1cm,0cm>:<0cm,1cm>::}
 	!{(0,-3) }*+{}="yu"
 	!{(0,4) }*+{{x_2}}="yo"
 	!{(-4,0) }*+{}="xl"
 	!{(3,0)}*+{{x_1}}="xr"
	!{(-3,-2) }*+{\circled{1}}="-3,-2"
	!{(-3,-1) }*+{\circled{2}}="-3,-1"
	!{(-3,0) }*+{\circled{2}}="-3,0"
	!{(-3,1) }*+{\circled{2}}="-3,1" 	
 	!{(-3,2) }*+{\circled{2}}="-3,2" 
 	!{(-3,3) }*+{\circled{2}}="-3,3" 
 	!{(-2,-1) }*+{\circled{1}}="-2,-1"
	!{(-2,0) }*+{\circled{2}}="-2,0"
	!{(-2,1) }*+{\circled{2}}="-2,1"
	!{(-2,2) }*+{\circled{2}}="-2,2" 	
 	!{(-2,3) }*+{\circled{2}}="-2,3" 
 	!{(-1,0) }*+{\circled{1}}="-1,0"
 	!{(-1,1) }*+{\circled{2}}="-1,1"
	!{(-1,2) }*+{\circled{2}}="-1,2" 	
 	!{(-1,3) }*+{\circled{2}}="-1,3" 
 	!{(0,1) }*+{\circled{1}}="0,1"
	!{(0,2) }*+{\circled{2}}="0,2" 	
 	!{(0,3) }*+{\circled{2}}="0,3" 
 	!{(1,2) }*+{\circled{1}}="1,2" 	
 	!{(1,3) }*+{\circled{2}}="1,3" 
 	!{(2,3) }*+{\circled{1}}="2,3"
 	!{(-3.75,-2)}*+{}="-3.75,-2"
 	!{(-3.75,-1)}*+{}="-3.75,-1"
 	!{(-3.75,0)}*+{}="-3.75,0"
 	!{(-3.75,1)}*+{}="-3.75,1"
 	!{(-3.75,2)}*+{}="-3.75,2"
 	!{(-3.75,3)}*+{}="-3.75,3"
 	!{(-3.,3.75)}*+{}="-3,3.75"
 	!{(-2.,3.75)}*+{}="-2,3.75"
 	!{(-1.,3.75)}*+{}="-1,3.75"
 	!{(0.,3.75)}*+{}="0,3.75"
 	!{(1.,3.75)}*+{}="1,3.75"
 	!{(2.,3.75)}*+{}="2,3.75"
 	!{(2.5,-2)}*+{}
 		"yu":"yo" "xl":"xr"
 		"-3,-2"-"-3,-1"-"-3,0"-"-3,1"-"-3,2"-"-3,3"
 		"-2,-1"-"-2,0"-"-2,1"-"-2,2"-"-2,3"
 		"-1,0"-"-1,1"-"-1,2"-"-1,3"
 		"0,1"-"0,2"-"0,3"
 		"1,2"-"1,3"
 		"-3,3"-"-2,3"-"-1,3"-"0,3"-"1,3"-"2,3"
 		"-3,2"-"-2,2"-"-1,2"-"0,2"-"1,2"
 		"-3,1"-"-2,1"-"-1,1"-"0,1"
 		"-3,0"-"-2,0"-"-1,0"
 		"-3,-1"-"-2,-1"
 		"-3.75,-2"-"-3,-2" "-3.75,-1"-"-3,-1" "-3.75,0"-"-3,0" 
 		"-3.75,1"-"-3,1" "-3.75,2"-"-3,2" "-3.75,3"-"-3,3"
 		"-3,3"-"-3,3.75" "-3,3"-"-3,3.75" "-3,3"-"-3,3.75" "-2,3"-"-2,3.75" "-1,3"-"-1,3.75" "0,3"-"0,3.75" "1,3"-"1,3.75" "2,3"-"2,3.75"}}
 		\caption{$H_\Z^2=-\frac{1}{2\Delta}h_0^{({\mathcal{X}}^2)}+W+2\cdot\idty$}
 		 \label{fig:1}
 		\end{figure}

Compare this with Figure \ref{fig:2} which shows the action of $\widetilde{H}^2$, i.e.\ after the change of coordinates. Again, the edges have a weight of $-\frac{1}{2\Delta}$ and the potential depicted is now $\widetilde{W}+2\cdot \idty$.
	\begin{figure}[htbp]
	\centering
	\fbox{ \xygraph{ 
	!{<0cm,0cm>;<1cm,0cm>:<0cm,1cm>::}
	!{(0,-1) }*+{}="yu"
 	!{(0,4) }*+{{n_1}}="yo"
 	!{(-4,0) }*+{}="xl"
 	!{(4,0)}*+{{x}}="xr"
 	!{(-3,1) }*+{\circled{1}}="-3,1"
 	!{(-2,1) }*+{\circled{1}}="-2,1"
 	!{(-1,1) }*+{\circled{1}}="-1,1"
 	!{(0,1) }*+{\circled{1}}="0,1"
 	!{(1,1) }*+{\circled{1}}="1,1"
 	!{(2,1) }*+{\circled{1}}="2,1"
 	!{(3,1) }*+{\circled{1}}="3,1"
 	!{(-3,2) }*+{\circled{2}}="-3,2"
 	!{(-2,2) }*+{\circled{2}}="-2,2"
 	!{(-1,2) }*+{\circled{2}}="-1,2"
 	!{(0,2) }*+{\circled{2}}="0,2"
 	!{(1,2) }*+{\circled{2}}="1,2"
 	!{(2,2) }*+{\circled{2}}="2,2"
 	!{(3,2) }*+{\circled{2}}="3,2"
 	!{(-3,3) }*+{\circled{2}}="-3,3"
 	!{(-2,3) }*+{\circled{2}}="-2,3"
 	!{(-1,3) }*+{\circled{2}}="-1,3"
 	!{(0,3) }*+{\circled{2}}="0,3"
 	!{(1,3) }*+{\circled{2}}="1,3"
 	!{(2,3) }*+{\circled{2}}="2,3"
 	!{(3,3) }*+{\circled{2}}="3,3"
 	!{(2.5,-0.5)}*+{}
 	!{(-3.75,1.75)}*+{}="-3.75,1.75"
 	!{(-3.75,2.75)}*+{}="-3.75,2.75"
 	!{(-3.75,3.75)}*+{}="-3.75,3.75"
 	!{(-2.75,3.75)}*+{}="-2.75,3.75"
 	!{(-1.75,3.75)}*+{}="-1.75,3.75"
 	!{(-0.75,3.75)}*+{}="-0.75,3.75"
 	!{(0.25,3.75)}*+{}="0.25,3.75"
 	!{(1.25,3.75)}*+{}="1.25,3.75"
 	!{(2.25,3.75)}*+{}="2.25,3.75"
 	!{(-3,3.75)}*+{}="-3,3.75"
 	!{(-2,3.75)}*+{}="-2,3.75"
 	!{(-1,3.75)}*+{}="-1,3.75"
 	!{(0,3.75)}*+{}="0,3.75"
 	!{(1,3.75)}*+{}="1,3.75"
 	!{(2,3.75)}*+{}="2,3.75"
 	!{(3,3.75)}*+{}="3,3.75"
 	!{(3.75,2.25)}*+{}="3.75,2.25"
 	!{(3.75,1.25)}*+{}="3.75,1.25"
 	"yu":"yo" "xl":"xr"
 	"-3,1"-"-3,2"-"-3,3"
 	"-2,1"-"-2,2"-"-2,3"
 	"-1,1"-"-1,2"-"-1,3"
 	"0,1"-"0,2"-"0,3"
 	"1,1"-"1,2"-"1,3"
 	"2,1"-"2,2"-"2,3"
 	"3,1"-"3,2"-"3,3"
 	"-2,1"-"-3,2" "-2,2"-"-3,3"
 	"-1,1"-"-2,2" "-1,2"-"-2,3"
 	"0,1"-"-1,2" "0,2"-"-1,3"
 	"1,1"-"0,2" "1,2"-"0,3"
 	"2,1"-"1,2" "2,2"-"1,3"
 	"3,1"-"2,2" "3,2"-"2,3"
 	"-3,1"-"-3.75,1.75"
 	"-3,2"-"-3.75,2.75"
 	"-3,3"-"-3.75,3.75"
 	"-3,3"-"-3,3.75"
 	"-2,3"-"-2,3.75"
 	"-1,3"-"-1,3.75"
 	"0,3"-"0,3.75"
 	"1,3"-"1,3.75"
 	"2,3"-"2,3.75"
 	"3,3"-"3,3.75"
 	"-2.75,3.75"-"-2,3"
 	"-1.75,3.75"-"-1,3"
 	"-0.75,3.75"-"0,3"
 	"0.25,3.75"-"1,3"
 	"1.25,3.75"-"2,3"
 	"2.25,3.75"-"3,3"
 	"3.75,2.25"-"3,3"
 	"3.75,1.25"-"3,2"
 		} } \caption{$\widetilde{H}^2=-\frac{1}{2\Delta}\widetilde{h}_0^{\left(\Z\times\N,\Gamma\right)}+\widetilde{W}+2\cdot\idty$}
 		\label{fig:2}
 		\end{figure}

Let us now give an explicit formula for the action of $h_0^{(\Z\times\N^m,\Gamma)}$ in terms of a ``hopping operator'', i.e.\ as a combination of translations. To this end, let $T$ be the right-shift on $\ell^2(\N)$,
\begin{equation}
(Tu)(n) = \left\{ \begin{array}{ll} u(n-1), & n\ge 2, \\ 0, & n=1. \end{array} \right.
\end{equation}
Observe that its adjoint $T^*$ is then just given by the (properly truncated) left-shift operator, i.e.\ $(T^*u)(n) = u(n+1)$ for all $n\in \N$. 
By $T_j$, $T_j^*$, $j=0,1,\ldots,m$, we denote the corresponding shifts which act on the $j$-th variable of functions in $\ell^2(\Z \times \N^m)$ with the understanding that $T_0$ and $T_0^*$ correspond to (non-truncated) shifts of the first coordinate in $\Z$.
Then $h_0^{(\Z\times\N^m,\Gamma)}$ is given by
\begin{equation} \label{eq:graphlaplace}
h_0^{(\Z\times \N^m, \Gamma)} = \sum_{j=0}^{m-1} (T_{j+1}^* T_j + T_j^* T_{j+1}) + T_m + T_m^*.
\end{equation}
Equation \eqref{eq:graphlaplace} could be shown by a direct calculation. However one could also be convinced that this equation holds true by a short picture-proof. For example, consider a particle hopping to its left (the second particle in the sketch).
\begin{align}
&\underset{x}{\bullet}\overbrace{\cdots\cdots\cdots}^{n_1}\bullet\overbrace{\cdots\cdots}^{n_2}\bullet\Rightarrow\\\Rightarrow&\underset{x}{\bullet}\overbrace{\cdots\cdots}^{n_1-1}\bullet\overbrace{\cdots\cdots\cdots}^{n_2+1}\bullet \nonumber
\end{align}
This operation is described by the operator $T_1T_2^*$. The fact that a particle cannot hop if the neighboring site is occupied is reflected by the truncation of the right-shift operator, i.e.\ $(Tf)(1)=0$. 

The interaction potential $(\widetilde{W} \varphi)(x,n_1,\ldots,n_m) = - \,\#(\{i:n_i=1\}) \,\varphi(x,n_1,\ldots,n_m)$ decomposes into the sum of $m$ two-particle interactions,
\begin{equation}
\widetilde{W}=\sum_{i=1}^m Q_i\:,
\end{equation}
with $Q$ from (\ref{eq:potential}) and $Q_i$ the operator acting as $Q$ on the $i$-th coordinate. Moreover, let us define $\widetilde{\N}=\N\setminus\{1\}$. For $j\in\{1,\dots,m\}$ let $P_j$ be the orthogonal projection onto the space of functions $\varphi\in\ell^2(\Z\times\N^m)$ with
 \begin{equation}
 \supp(\varphi)\subset\Z\times\N^{j-1}\times\widetilde{\N}\times\N^{m-j}\:.
 \end{equation} 
Then it is easy to observe that $Q_i=P_i-\idty$, which enables us to write
\begin{equation}
\widetilde{W}+N\cdot\idty=\idty+\sum_{i=1}^mP_i\:.
\end{equation}
 
\subsection{Fourier transformation in $x$-direction}

Observe that $\widetilde{W}$ only depends on the relative distances between the particles and therefore is independent of $x$. Thus $\widetilde{H}^N$ is invariant under translations in $x$-direction and, analogously to \cite{NSS}, we can reduce the dimension by one with a Bloch expansion. In the new coordinates this simply becomes a Fourier transform in the $x$ variable.

For this purpose, we identify $\ell^2(\Z\times\N^m)=\ell^2(\Z; \ell^2(\N^m))$ and apply the Fourier transform
\begin{equation}
\mathcal{F}: \ell^2(\Z;\ell^2(\N^m))\rightarrow L^2([-\pi,\pi); \ell^2(\N^m))\:,
\end{equation}
given by
\begin{equation}
\big(\mathcal{F}g)(\vartheta, n_1,\dots,n_m)=\frac{1}{\sqrt{2\pi}}\sum_{x\in\Z}e^{i x\vartheta} g(x,n_1,\dots,n_m) 
\end{equation}
and inverse
\begin{equation}
(\mathcal{F}^{-1}f)(x,n_1,\dots,n_m) = \frac{1}{\sqrt{2\pi}}\int_{[-\pi,\pi)}e^{-ix\vartheta} f(\vartheta,n_1,\dots,n_m)\text{d}\vartheta\:.
\end{equation}

By translation invariance, the operator $\widehat{H}^N:=\mathcal{F}\widetilde{H}^N\mathcal{F}^{-1}$ acts as a generalized multiplication operator (or direct integral) on $L^2([-\pi,\pi); \ell^2(\N^m))=\int^\oplus_{[-\pi,\pi)}\ell^2(\N^m)\text{d}\vartheta$,
\begin{equation}
\big(\widehat{H}^Nf\big)(\vartheta,n_1,\dots,n_m)=\left(\widehat{H}^N(\vartheta)f(\vartheta)\right)(n_1,\dots,n_m)\:,
\end{equation}
where $\widehat{H}^N(\vartheta)$ has domain $\ell^2(\N^m)$.
The concrete action of $\widehat{H}^N(\vartheta)$ can be obtained from equation \eqref{eq:graphlaplace} by observing that the translation operators $T_0$ and $T_0^*$ are transformed into multiplication operators by a complex phase, $\mathcal{F}T_0^*\mathcal{F}^{-1}=e^{-i\vartheta}$ and $\mathcal{F}T_0\mathcal{F}^{-1}=e^{i\vartheta}$, while $T_j$ and $T_j^*$ for $j\ge 1$ commute with $\mathcal{F}$. 

Thus
\begin{equation} \label{eq:fiberoperator}
\widehat{H}^N(\vartheta)  = -\frac{1}{2\Delta} \left(e^{i\vartheta} T_1^* + e^{-i\vartheta} T_1 + \sum_{j=1}^{m-1} (T_{j+1}^* T_j + T_j^* T_{j+1}) + T_m^* + T_m \right) + \sum_{j=1}^m P_j + \idty.
\end{equation}

For $N=2$ these operators take the form of Jacobi matrices
\begin{equation} \label{eq:BlochN2}
\widehat{H}^2(\vartheta) = -\frac{1}{2\Delta} \left( (1+e^{i\vartheta}) T_1^* +(1+e^{-i\vartheta}) T_1 \right) + P_1 + \idty
\end{equation}
on $\ell^2(\N)$. Figure \ref{fig:3} visualizes the operators $\widehat{H}^3(\vartheta)$ on $\ell^2(\N^2)$.
\begin{figure}[htbp]
\centering
 \fbox{ \xygraph{ 
	!{<0cm,0cm>;<1cm,0cm>:<0cm,1cm>::}
 	!{(0,-1) }*+{}="yu"
 	!{(0,5.5) }*+{{n_2}}="yo"
 	!{(-1,0) }*+{}="xl"
 	!{(5.5,0)}*+{{n_1}}="xr"
 	!{(1,1) }*+{\circled{1}}="1,1"
 	!{(2,1) }*+{\circled{2}}="2,1"
 	!{(3,1) }*+{\circled{2}}="3,1"
 	!{(4,1) }*+{\circled{2}}="4,1"
 	!{(1,2) }*+{\circled{2}}="1,2"
 	!{(2,2) }*+{\circled{3}}="2,2"
 	!{(3,2) }*+{\circled{3}}="3,2"
 	!{(4,2) }*+{\circled{3}}="4,2"
 	!{(1,3) }*+{\circled{2}}="1,3"
 	!{(2,3) }*+{\circled{3}}="2,3"
 	!{(3,3) }*+{\circled{3}}="3,3"
 	!{(4,3) }*+{\circled{3}}="4,3"
 	!{(1,4) }*+{\circled{2}}="1,4"
 	!{(2,4) }*+{\circled{3}}="2,4"
 	!{(3,4) }*+{\circled{3}}="3,4"
 	!{(4,4) }*+{\circled{3}}="4,4"
 	!{(1,5) }*+{}="1,5"
 	!{(2,5) }*+{}="2,5"
 	!{(3,5) }*+{}="3,5"
 	!{(4,5) }*+{}="4,5"
 	!{(5,1) }*+{}="5,1"
 	!{(5,2) }*+{}="5,2"
 	!{(5,3) }*+{}="5,3"
 	!{(5,4) }*+{}="5,4"
 	!{(5,-0.5) }*+{}
 	"yu":"yo" "xl":"xr"
 	"1,1":"2,1":"3,1":"4,1":"5,1"
 	"1,2":"2,2":"3,2":"4,2":"5,2"
 	"1,3":"2,3":"3,3":"4,3":"5,3"
 	"1,4":"2,4":"3,4":"4,4":"5,4"
 	"1,1"-"1,2"-"1,3"-"1,4"-"1,5"
 	"2,1"-"2,2"-"2,3"-"2,4"-"2,5"
 	"3,1"-"3,2"-"3,3"-"3,4"-"3,5"
 	"4,1"-"4,2"-"4,3"-"4,4"-"4,5"
 	"2,1"-"1,2"
 	"3,1"-"2,2"-"1,3"
 	"4,1"-"3,2"-"2,3"-"1,4"
 	"4,2"-"3,3"-"2,4"
 	"4,3"-"3,4"
 	"2,4"-"1,5"  "3,4"-"2,5" "4,4"-"3,5"
 	"4,4"-"5,3" "4,3"-"5,2" "4,2"-"5,1"
 	} }
 	\caption{The action of $\widehat{H}^3(\vartheta)$.}
 	\label{fig:3}
 	\end{figure}
 	Again, the edges have a weight of $-\frac{1}{2\Delta}$ and the vertices are labeled with the values of $P_1+P_2+\idty$, which now are $1$, $2$ or $3$. As a special feature, the horizontal edges of this graph are directed now. Moving along these edges one picks up a phase $e^{\pm i \vartheta}$, depending on the direction.

\subsection{Computation of the droplet spectrum}

In this section, we will use the Bethe ansatz to determine an eigenvalue $E_N(\vartheta)$ of $\widehat{H}^N(\vartheta)$. As $\vartheta$ varies, these eigenvalues yield the droplet band of $H_{\Z}^N$ introduced in (\ref{eq:dropletband}) below.

\begin{prop} \label{prop:eigval}
For $m\geq2$ and $\vartheta\in(-\pi,\pi)$ let $\varphi_\vartheta$ be a function from $\N^m$ to $\C$, which is defined as
\begin{equation} \label{eq:bethe}
\varphi_\vartheta(n_1,n_2,\dots,n_m)=a_1(\vartheta)^{n_1}\cdot a_2(\vartheta)^{n_2}\cdot\,\dots\,\cdot a_m(\vartheta)^{n_m}=\prod_{i=1}^m a_i(\vartheta)^{n_i}\:,
\end{equation}
where the numbers $\{a_i(\vartheta)\}_{i=1}^m$ are determined by the following Jacobi matrix equation:
\begin{equation} 
\begin{pmatrix}
2\Delta& -1\\
-1& 2\Delta&-1\\
&-1&\ddots&\ddots\\
&&\ddots&\ddots&-1\\
&&&-1&2\Delta&-1\\
&&&&-1&2\Delta
\end{pmatrix}
\begin{pmatrix}
a_1(\vartheta)\\a_2(\vartheta)\\\vdots\\\vdots\\a_{m-1}(\vartheta)\\a_m(\vartheta)
\end{pmatrix}
=
\begin{pmatrix}
e^{-i\vartheta}\\0\\\vdots\\\vdots\\0\\1
\end{pmatrix}\:.
\label{eq:matriks}
\end{equation}

Then

(a) there exists a number $E_N(\vartheta)$ such that $\varphi_\vartheta$ is a solution to the equation
\begin{equation} \label{eq:eigwertgl}
\widehat{H}^N(\vartheta)\varphi_\vartheta=E_N(\vartheta)\varphi_\vartheta\:,
\end{equation}
(b) $\varphi_\vartheta$ is non-trivial and square-summable, i.e.\ $0\neq\varphi_\vartheta\in\ell^2(\N^m)$, and therefore $E_N(\vartheta)$ is an eigenvalue of $\widehat{H}^N(\vartheta)$.
\end{prop}

Here we have excluded the cases $m=0$ and $m=1$ as well as $\vartheta=-\pi$, which will be commented on at the end of this section. Note that (\ref{eq:eigwertgl}) is initially understood as a difference equation, in which $\widehat{H}^N(\vartheta)$ acts on general functions $\varphi:\N^m\to \C$ formally via (\ref{eq:fiberoperator}). That this can also be understood as a Hilbert space equation follows from (b).

\begin{proof}
(a) We show this part by an explicit construction of the solution. First, observe that for $\Delta>1$, equation \eqref{eq:matriks} has a unique solution, since the matrix on the left side of this equation is positive definite in this case. This means that $\varphi_\vartheta$ is well-defined.

The proof of part (b) below will also show that $a_i(\vartheta)\neq 0$ for all $i=1,\ldots,m$.

We have to distinguish two cases.

\begin{itemize}

\item First case: All $n_i>1$. In this case, all of the shift operators of the form $T_jT^*_{j+1}$ as well as all projections $P_j$ yield a contribution and we obtain from (\ref{eq:fiberoperator}) and (\ref{eq:bethe}) that
\begin{align}
\left(\widehat{H}^N(\vartheta)\varphi_\vartheta\right)(n_1,\dots,n_m)&= \\ \nonumber
=\Bigg[-\frac{1}{2\Delta}\Bigg(e^{i\vartheta}a_1+e^{-i\vartheta}\frac{1}{a_1}&+\frac{a_2}{a_1}+\frac{a_1}{a_2}+\dots\\ \nonumber \dots+&\frac{a_m}{a_{m-1}} + \frac{a_{m-1}}{a_m}+a_m + \frac{1}{a_m}\Bigg)+(m+1)\Bigg]\varphi_\vartheta(n_1,\dots,n_m)\:,
\end{align}
where we have omitted the argument $\vartheta$.

\item Second case: For some $i\in I\subset\{1,\dots,m\}$ we have that $n_i=1$. This means that all terms in \eqref{eq:fiberoperator}, which are operators of the form $T^*_{i-1}T_i$ and $T^*_{i+1}T_i$ for $i\in I$ as well as all projections $\{P_i\}_{i\in I}$ give zero contribution in this case. Thus for $n_i=1$, terms of the form 
\begin{align}
-\frac{1}{2\Delta}\left(e^{-i\vartheta}\frac{1}{a_1}+\frac{a_2}{a_1}\right)+1 & \quad \mbox{for $i=1$}, \label{eq:term1} \\
-\frac{1}{2\Delta}\left(\frac{a_{i-1}}{a_i}+\frac{a_{i+1}}{a_i}\right)+1 & \quad\mbox{for $1<i<m$}, \label{eq:term2} \\
-\frac{1}{2\Delta}\left(\frac{a_{m-1}}{a_m}+\frac{1}{a_m}\right)+1 & \quad \mbox{for $i=m$} \label{eq:term3}
\end{align}
will not occur in the prefactor of $\varphi_\vartheta$ as it did in the first case.
The condition that equation \eqref{eq:eigwertgl} holds for the same number $E_N(\vartheta)$ at all sites $(n_1,\ldots,n_m) \in \N^m$ therefore requires that all the $m$ terms in (\ref{eq:term1}), (\ref{eq:term2}) and (\ref{eq:term3}) vanish. But this is easily seen to be equivalent to (\ref{eq:matriks}).
\end{itemize}

(b) The proof of this claim mainly consists in explicitly solving equation \eqref{eq:matriks} and showing that $0<|a_i(\vartheta)|<1$ for all $i$ and $\vartheta$, which implies square-summability. For this purpose, multiply equation \eqref{eq:matriks} by $e^{i\vartheta/2}$ and define 
\begin{equation} \label{eq:relabel}
b_k(\vartheta):=e^{i\vartheta/2}a_{k+L+1}(\vartheta)
\end{equation}
for $k\in\{-L,\dots,L\}$, where $L:=\frac{m-1}{2}$. Observe that the $k$'s are half-integer for even $m$ and integer for odd $m$. Now, equation \eqref{eq:matriks} looks like
\begin{equation} 
\begin{pmatrix}
2\Delta& -1\\
-1& 2\Delta&-1\\
&-1&\ddots&\ddots\\
&&\ddots&\ddots&-1\\
&&&-1&2\Delta&-1\\
&&&&-1&2\Delta
\end{pmatrix}
\begin{pmatrix}
b_{-L}(\vartheta)\\b_{-(L-1)}(\vartheta)\\\vdots\\\vdots\\b_{L-1}(\vartheta)\\b_L(\vartheta)
\end{pmatrix}
=
\begin{pmatrix}
e^{-i\vartheta/2}\\0\\\vdots\\\vdots\\0\\e^{i\vartheta/2}
\end{pmatrix}\:.
\label{eq:matriks3}
\end{equation}
Again we are going to omit the argument $\vartheta$. Using only the homogeneous part of (\ref{eq:matriks3}), i.e.\ the rows corresponding to $-(L-1) \le k \le L-1$ shows that the solution is of the form
\begin{equation}
b_n=c_+e^{\rho n}+c_-e^{-\rho n}\:, \qquad -L\leq n\leq L\:,
\end{equation}
where $\pm\rho=\log(\Delta\pm\sqrt{\Delta^2-1})$ are the logarithms of the eigenvalues of the transfer matrix 
\begin{equation}
T=\begin{pmatrix} 0&1\\-1&2\Delta\end{pmatrix}\:.
\label{eq:transfer}
\end{equation}
As \eqref{eq:matriks3} is symmetric up to complex conjugation we have $b_n=\overline{b_{-n}}$ for all $n$, and thus $c_+=\overline{c_-}$. Therefore we may write
\begin{equation} \label{eq:grundzahl}
b_n=c_+e^{\rho n}+\overline{c_+}e^{-\rho n}=2\text{Re\:}(c_+)\cosh(\rho n)+2i\text{Im\:}(c_+)\sinh(\rho n)
\end{equation}
and
\begin{equation}
\left|b_n\right|^2=4\left(\text{Re\:}(c_+)\right)^2\cosh(\rho n)^2+4\left(\text{Im\:}(c_+)\right)^2\sinh(\rho n)^2\:,
\label{eq:schluss}
\end{equation}
which is monotonically increasing in $|n|$.
The constant $c_+$ can be determined by the boundary condition given by the first line of the matrix equation \eqref{eq:matriks3}:
\begin{equation}
2\Delta \left(c_+ e^{\rho L}+\overline{c_+}e^{-\rho L}\right)-\left(c_+e^{\rho(L-1)}+\overline{c_+}e^{-\rho(L-1)}\right)=e^{-i\vartheta/2}\;,
\end{equation}
which yields
\begin{equation}
c_+=\frac{\cos\vartheta/2}{4\Delta\cosh(\rho L)-2\cosh(\rho(L-1))}-i\frac{\sin\vartheta/2}{4\Delta\sinh(\rho L)-2\sinh(\rho(L-1))}\:.
\label{eq:cplus}
\end{equation}
To give an upper bound for the $|b_n|$'s we use the already observed fact that $|b_n|$ is monotonically increasing in $|n|$. It therefore suffices to estimate $|b_L|$:
\begin{equation} \label{eq:bLabs}
|b_L|^2=\frac{4\cos(\vartheta/2)^2\cosh(\rho L)^2}{(4\Delta\cosh(\rho L)-2\cosh(\rho(L-1)))^2}+\frac{4\sin(\vartheta/2)^2\sinh(\rho L)^2}{(4\Delta\sinh(\rho L)-2\sinh(\rho(L-1)))^2}\;.
\end{equation}
As we have $2\cosh(\rho(L-1))\leq2\Delta\cosh(\rho L)$ and $2\sinh(\rho(L-1))\leq2\Delta\sinh(\rho L)$ it is clear that
\begin{align}
4\Delta\cosh(\rho L)-2\cosh(\rho(L-1))&\geq2\Delta\cosh(\rho L)\;,\\
4\Delta\sinh(\rho L)-2\sinh(\rho(L-1))&\geq2\Delta\sinh(\rho L)\:.
\end{align}
This and (\ref{eq:bLabs}) readily imply that $|b_L|^2 \le  1/\Delta^2 < 1$. 

Moreover, observe that for $\vartheta\neq\pm\pi$ we have that $\cos(\vartheta/2)^2>0$ and therefore $(\text{Re\:}(c_+))^2>0$, which implies $|b_n|>0$ for all $n\in\{-L,\dots,L\}$ according to equation \eqref{eq:schluss}. This completes the proof of (b).
\end{proof}

Using equations \eqref{eq:relabel}, \eqref{eq:grundzahl} and \eqref{eq:cplus}, we are now able to give an explicit formula for $a_k(\vartheta)$. Observe that the definition $\rho=\log(\Delta+\sqrt{\Delta^2-1})$ is equivalent to $\Delta=\cosh\rho$, which allows to simplify
\begin{align}
4\Delta\cosh(\rho L)-2\cosh(\rho(L-1))&=2\cosh((L+1)\rho)\;,\\
4\Delta\sinh(\rho L)-2\sinh(\rho(L-1))&=2\sinh((L+1)\rho)\:.
\end{align}
This and \eqref{eq:grundzahl}, \eqref{eq:cplus} yield 
\begin{equation} \label{eq:eigfunc}
b_k(\vartheta)=\frac{\cos({\vartheta}/{2})\cosh(k\rho)}{\cosh((L+1)\rho)}+i\frac{\sin(\vartheta/2)\sinh(k\rho)}{\sinh((L+1)\rho)}
\end{equation}
and with \eqref{eq:relabel},
\begin{equation} \label{eq:resub}
a_k(\vartheta) = e^{-i\vartheta/2}\left(\frac{\cos({\vartheta}/{2})\cosh\left(\left(k-\frac{N}{2}\right)\rho\right)}{\cosh\left(\frac{N}{2}\rho\right)}+i\frac{\sin(\vartheta/2)\sinh\left(\left(k-\frac{N}{2}\right)\rho\right)}{\sinh\left(\frac{N}{2}\rho\right)}\right) 
\end{equation}
for $k\in\{1,\dots,N-1\}$, where we have resubstituted $m+1=N$.

Finally, let us compute the eigenvalue $E_N(\vartheta)$. This is most easily done by evaluating (\ref{eq:eigwertgl}) at the site $n_1=\ldots = n_m =1$, because at this site $\widehat{H}^N(\vartheta)$ acts simply as $-\frac{1}{2\Delta}\left(e^{i\vartheta}T_1^*+T_m^*\right)+\idty$. This yields
\begin{equation} \label{eq:eigval}
\left(\widehat{H}_N(\vartheta)\varphi_\vartheta\right)(1,\dots,1) = \left[-\frac{1}{2\Delta}(e^{i\vartheta}a_1(\vartheta)+a_m(\vartheta))+1\right] \varphi_\vartheta(1,\dots,1) = E_N(\vartheta) \varphi_\vartheta(1,\dots,1)\;.
\end{equation}
With the values from (\ref{eq:resub}) and after some simplification we arrive at
\begin{equation} \label{eq:dropleteigval}
E_N(\vartheta)=\tanh(\rho)\cdot\frac{\cosh(N\rho)-\cos(\vartheta)}{\sinh(N\rho)}\:,
\end{equation}
where, as before, $\rho$ is the positive number such that $\cosh(\rho)=\Delta$. 

The reason that $\vartheta=-\pi$ had to be excluded in the above construction is that in this case the eigenfunction given by (\ref{eq:bethe}) may vanish identically. Nevertheless, the formula (\ref{eq:dropleteigval}) gives an eigenvalue of $\widehat{H}^N(\vartheta)$ also for $\vartheta=-\pi$. This is seen as follows:

\begin{itemize}

\item If $N$ is odd, than nothing changes in the above construction, as it still holds that $0< |a_k(\vartheta)|$, $1\le k \le N-1$ (note that the imaginary part of the second factor in (\ref{eq:resub}) is non-zero).

\item If $N$ is even, then (\ref{eq:resub}) yields $a_{N/2}(-\pi) = 0$ so that $\varphi_{-\pi}$ given by (\ref{eq:bethe}) would vanish. In this case the correct eigenfunction is given by
\begin{equation}
\varphi_{-\pi}(n_1,\ldots,n_{N-1}) = \prod_{k=1}^{N/2-1} a_k(-\pi)^{n_k} \cdot \delta_1(n_{N/2}) \cdot \prod_{k=N/2+1}^{N-1} a_k(-\pi)^{n_k}.
\end{equation}
Note that this eigenfunction is supported on the hypersurface $n_{N/2}=1$ in $\N^{N-1}$.

\end{itemize}

The operators $\widehat{H}^N(\vartheta)$ are a Bloch decomposition for $\widehat{H}^N$. Thus the union of their eigenvalues $E_N(\vartheta)$ is contained in the spectrum of $\widehat{H}^N$ (and thus of $H_{\Z}^N$),
\begin{equation} \label{eq:dropletband}
\delta_N := \left[\tanh(\rho)\cdot\frac{\cosh(N\rho)-1}{\sinh(N\rho)},\tanh(\rho)\cdot\frac{\cosh(N\rho)+1}{\sinh(N\rho)}\right] \subset \sigma(H_{\Z}^N) \;.
\end{equation}

We refer to $\delta_N$ as the {\it droplet band} or {\it droplet spectrum} of $H_{\Z}^N$. Let us discuss the reasons for this terminology. Via the Bloch decomposition the eigenfunctions $\varphi_{\vartheta}$ of $\widehat{H}^N(\vartheta)$ yield generalized eigenfunctions
\begin{equation}
\tilde{f}_{\vartheta}(x,n_1,\ldots,n_m) := e^{-i\vartheta x} \varphi_{\vartheta}(n_1,\ldots,n_m)
\end{equation}
for $\widetilde{H}^N$, i.e.\
\begin{equation}
\widetilde{H}^N \tilde{f}_{\vartheta} = E_N(\vartheta) \tilde{f}_{\vartheta}
\end{equation}
in the sense of a difference equation on $\Z \times \N^m$. After undoing the change of variables $V_0$ we get generalized eigenfunctions
\begin{equation} \label{eq:genefs}
f_{\vartheta}(x_1,\ldots,x_N) = e^{-i\vartheta x_1} \varphi_{\vartheta}(x_2-x_1, x_3-x_2, \ldots, x_N-x_{N-1})
\end{equation}
of $H_{\Z}^N$ corresponding to energies in the droplet band. These functions attain their maximal absolute value if and only if $x_{i+1}-x_i=1$ for all $i$ and decay exponentially in $x_{i+1}-x_i$ for all $i$. In our $N$-body interpretation this means that all $N$ particles (or down spins in the XXZ chain) are tightly packed together, forming a droplet, with corrections decaying exponentially in all the particle distances.

While we assumed $m\ge 2$ and thus $N\ge 3$ in Theorem~\ref{prop:eigval}, the formula (\ref{eq:dropletband}) holds true also for $N=1$ and $N=2$. This is trivial for $N=1$ as $H_{\Z}^1 = - \frac{1}{2\Delta} h_0^{(\Z)} + \idty$, whose spectrum is 
\begin{equation} \label{eq:delta1}
[1-1/\Delta, 1+1/\Delta] = \delta_1.
\end{equation}
Thus in this case $\delta_1$ covers the entire spectrum, although the interpretation as droplet spectrum does not make sense in the presence of a single particle.

For $N=2$ we have the Bloch decomposition (\ref{eq:BlochN2}), but defining the coefficients in the Bethe ansatz by (\ref{eq:matriks}) does not make sense in this case. Instead, observe by direct inspection that, for $\vartheta \in (-\pi,\pi)$,
\begin{equation}
\varphi_{\vartheta}(n) = \left( \frac{1+e^{i\vartheta}}{2\Delta} \right)^n
\end{equation}
is an exponentially decaying eigenfunction of $\widehat{H}^2(\vartheta)$ to the eigenvalue $E_2(\vartheta) = 1- \frac{1+\cos(\vartheta)}{2\Delta^2}$. For $\vartheta=-\pi$ we have the eigenfunction $\varphi_{-\pi}(n) = \delta_1(n)$ with eigenvalue $1$ (note that $\widehat{H}^2(-\pi)$ is the diagonal matrix $P_1+ \idty$). The union of these eigenvalue over $-\pi \le \vartheta < \pi$ is the droplet band
\begin{equation} \label{eq:delta2}
\left[ 1- \frac{1}{\Delta^2}, 1 \right] = \delta_2
\end{equation}
with corresponding droplet generalized eigenfunctions $f_{\vartheta}(x_1,x_2) = e^{-i\vartheta x_1} \varphi_{\vartheta}(x_2-x_1)$, as in (\ref{eq:genefs}).

\begin{cor} \label{cor:specmin}
The minimum of $\sigma(H_{\Z}^N)$ is given by
\begin{equation} \label{eq:specminN}
\min\sigma(H_{\Z}^N)=E_N(0)=\tanh(\rho)\cdot\frac{\cosh(N\rho)-1}{\sinh(N\rho)}\:.
\end{equation}
\end{cor}

\begin{proof}
For $\vartheta=0$ look at the eigenfunction $\varphi_0$. From equation \eqref{eq:resub} we see that
\begin{equation}
a_k(0)=\frac{\cosh\left(\left(k-\frac{N}{2}\right)\rho\right)}{\cosh\left(\frac{N}{2}\rho\right)}>0\quad \mbox{for all $N$, $k$}
\end{equation}
and therefore $\varphi_0>0$. From (\ref{eq:genefs}) we see that $H_{\Z}^N$ has positive generalized eigenfunction $f_0(x_1,\ldots,x_N) = \varphi_0(x_2-x_1, \ldots, x_N-x_{N-1})$ to $E_N(0)$.  By the Allegretto-Piepenbrink Theorem (see Theorem~3.1 of \cite{HaeselerKeller} for a version for Laplacians on general graphs)  this implies that 
\begin{equation}
E_N(0)\leq\min\sigma\left(H_{\Z}^N\right)\:.
\end{equation}
But we already know that $E_N(0)\in\sigma\left(H_{\Z}^N \right)$, which means that $E_N(0)$ has to be the spectral minimum.
\end{proof}

\section{An HVZ type theorem} \label{sec:HVZ}

Our next goal is a more thorough description of the spectrum of the restrictions $H_{\Z}^N$ of the free XXZ chain to the $N$-particle subspace. Having identified these operators as discrete $N$-body Schr\"odinger operators in Section~\ref{sec:Schrodinger} will allow us to do this in terms of cluster decompositions of the $N$-particle system and a related HVZ type theorem. In this section we will consider a more general class of $N$-body operators, followed by an application to our specific model in Section~\ref{sec:gap} below.

While HVZ theorems and the underlying methods are very classical and it is unlikely that our particular result is entirely new, we have not found a proper reference for this in the context of discrete Schr\"odinger operators.  We start with a result for the ``full'' $N$-particle Hilbert space, i.e.\ without symmetry restrictions, and then show that a similar result holds for the symmetric and anti-symmetric subspaces. As we are only dealing with bounded Hamiltonians in $\ell^2$-spaces, the proofs are quite straightforward and considerably simpler than the technically more involved arguments in proofs of HVZ theorems for classical (continuum) Schr\"odinger operators, see e.g.\ Section~XIII.5 in \cite{ReedSimonIV}, where also many references can be found, including to the original works by Hunziker, Van Winter and Zhislin.

\subsection{An HVZ type theorem} \label{subsec:HVZ}

As one-particle operator we choose an arbitrary bounded, self-adjoint and translation invariant operator $h$ in $\ell^2(\Z)$, i.e.\
\begin{equation}
[h,T_x] = 0 \quad \mbox{for all $x\in \Z$},
\end{equation}
where $(T_x f)(y) = f(y-x)$ for all $y\in \Z$.

For the interaction potential $Q:\{0,1,2,\ldots\} \to \R$ we assume that
\begin{equation}
\lim_{r\to\infty} Q(r) =0.
\end{equation}

The interacting $N$-particle Hamiltonian on $\ell^2(\Z^N) = \bigotimes^N \ell^2(\Z)$ is given by
\begin{equation} \label{eq:NpartHam}
H_N :=\sum_{i=1}^N h_i + \sum_{1\leq i<j\leq N}Q_{ij}\:,
\end{equation}
where 
\begin{equation}
h_i:=\idty\otimes\idty\otimes\dots\otimes\idty\otimes \underbrace{h}_{i^{th} entry} \otimes \idty\otimes\dots\otimes \idty
\end{equation}
and 
\begin{equation}
Q_{ij}(z)=Q(|z_i-z_j|),\quad\text{for $z\in \Z^N$}\;.
\end{equation}

\begin{prop} \label{prop:HVZ}
For any $N\in \Z$ let $H_N$ be the $N$-particle Hamiltonian defined by (\ref{eq:NpartHam}) with assumptions as above. Then, for any $m\in\{1,2,\dots, N-1\}$, we have
\begin{equation}
\sigma(H_m)+\sigma(H_{N-m})\subset\sigma(H_N)\:.
\end{equation}
\label{prop1}
\end{prop}

\begin{proof}
Let $\lambda \in \sigma(H_m)$ and $\mu \in \sigma(H_{N-m})$. As $H_m$ and $H_{N-m}$ are bounded, there are finitely supported and normalized Weyl sequences $\psi_n \in \ell^2(\Z^m)$ and $\varphi_n \in \ell^2(\Z^{N-m})$ to $\lambda$ and $\mu$, respectively, i.e.\
\begin{equation}
\|(H_m-\lambda)\psi_n\| \to 0, \quad \|(H_{N-m}-\mu)\varphi_n\| \to 0 \quad \mbox{as $n\to\infty$},
\end{equation}
see, e.g., Theorem~7.22 in \cite{Weidmann}. Thus there are finite integers $\sigma_n^{min/max}$ and $\tau_n^{min/max}$ such that
\begin{equation}
\supp \psi_n \subset [\sigma_n^{min}, \sigma_n^{max}]^m, \quad \supp \varphi_n \subset [ \tau_n^{min}, \tau_n^{max}]^{N-m}\;.
\end{equation}

To $\varepsilon>0$ let $R\in \N$ be such that 
\begin{equation} \label{eq:Wsmall}
|Q(r)| < \frac{\varepsilon}{3m(N-m)} \quad \mbox{for all $r\ge R$}\;.
\end{equation}
For all $n\in \N$ choose $c_n \in \Z$ such that
\begin{equation} \label{eq:distsupp}
\dist \left( [ \sigma_n^{min}, \sigma_n^{max}], [\tau_n^{min} + c_n, \tau_n^{max} + c_n] \right) > R\;.
\end{equation}
Write $z=(x,y) \in \Z^N$ with $x\in \Z^m$, $y\in \Z^{N-m}$ and let
\begin{equation}
\chi_n(z) := \psi_n(x) (T_{(c_n,\ldots,c_n)} \varphi_n)(y)\;
\end{equation}
where $T_{(c_n,\ldots,c_n)}$ denotes a multi-dimensional translation along the diagonal in $\Z^{N-m}$. Note that the $\chi_n$ are normalized in $\ell^2(\Z^N)$. 

We may write
\begin{equation}
H_N = H_m \otimes \idty + \idty \otimes H_{N-m} + \sum_{i\le m,\; j\ge m+1} Q_{ij}\;,
\end{equation}
which, together with $[T_{(c_n,\ldots,c_n)}, H_{N-m}]=0$, yields
\begin{equation}
\|(H_N - (\lambda + \mu)) \chi_n \|  \le  \|(H_m - \lambda) \psi_n\| + \|(H_{N-m} - \mu) \varphi_n\| + \sum_{i\le m, \,j\ge m+1} \| Q_{ij} \chi_n\|\;. 
\end{equation}
By (\ref{eq:Wsmall}) and (\ref{eq:distsupp}) each term in the sum on the right hand side is less than $\varepsilon/(3m(N-m))$, and thus the sum less than $\varepsilon/3$, uniformly in $n$. For $n$ sufficiently large the other two terms on the right hand side will also be less than $\varepsilon/3$ and thus $\|(H_N - (\lambda + \mu)) \chi_n\| < \varepsilon$, proving that $\lambda + \mu \in \sigma(H_N)$.
\end{proof}

\subsection{Extension to symmetric and antisymmetric subspaces}

In Proposition~\ref{prop:Schrodinger} we have shown that, when considered on the $N$-particle subspace, the infinite XXZ spin chain is equivalent to the restriction of a discrete $N$-body Schr\"odinger operator to the \emph{antisymmetric} subspace $\ell_a^2(\Z^N)$of $\ell^2(\Z^N)$. In order to apply the result of Proposition~\ref{prop1} to our specific problem, we therefore have to extend its result to antisymmetric subspaces. This will also work for symmetric subspaces.

Let $\ell^2_a(\Z^N)$ be the antisymmetric subspace of $\ell^2(\Z^N)$ as introduced in Section~\ref{sec:Schrodinger} above, and also consider the symmetric subspace $\ell^2_s(\Z^N)$ given by all $f\in \ell^2(\Z^N)$ such that 
\begin{equation}
f(x_{\sigma(1)}, \ldots, x_{\sigma(N)}) = f(x_1,\ldots,x_N)
\end{equation}
for all $\sigma \in S_N$. The symmetric and antisymmetric subspaces are invariant under $H_N$. Let $H_N^{(s)}$ and $H_N^{(a)}$ denote the restrictions of $H_N$ to $\ell^2_s(\Z^N)$ and $\ell^2_a(\Z^N)$, respectively. 

\begin{prop} \label{prop:HVZsub}
For all $m \in \{1,\ldots,N-1\}$ we have
\begin{equation}
\sigma(H_m^{(s)}) + \sigma(H_{N-m}^{(s)}) \subset \sigma(H_N^{(s)})
\end{equation}
and
\begin{equation}
\sigma(H_m^{(a)}) + \sigma(H_{N-m}^{(a)}) \subset \sigma(H_N^{(a)})\;.
\end{equation}
\label{prop:symm}
\end{prop}

\begin{proof}
For $i\in \{m, N-m, N\}$ and $\sigma \in S_i$ consider the unitary permutation operators $U_{\sigma}$ on $\ell^2(\Z^i)$ defined by
\begin{equation}
(U_{\sigma} \psi)(x_1,\ldots,x_i) = \psi(x_{\sigma(1)},\ldots, x_{\sigma(i)}).
\end{equation}
Then the orthogonal projections $P_s^{(i)}$ onto $\ell^2_s(\Z^i)$ and $P_a^{(i)}$ onto $\ell^2_a(\Z^i)$ in $\ell^2(\Z^i)$ are given by
\begin{equation}
P_s^{(i)} = \frac{1}{i!} \sum_{\sigma \in S_i} U_{\sigma}, \quad P_a^{(i)} = \frac{1}{i!} \sum_{\sigma \in S_i} \mbox{\rm sgn}(\sigma) U_{\sigma}\;.
\end{equation}
With this one can formally check that $[H_N, P_*^{(N)}]=0$, $* \in \{s, a\}$, i.e.\ the invariance of the symmetric and antisymmetric subspaces under $H_N$ claimed above (which also could be considered as physically evident).

Observe that
\begin{align}
P_s^{(N)} & = \frac{m!(N-m)!}{N!} P_s^{(m)} \otimes P_s^{(N-m)} +\frac{1}{N!} \sum_{\sigma\in S_N'} U_\sigma\label{eq:vollsymm}\\
P_a^{(N)} & = \frac{m!(N-m)!}{N!} P_a^{(m)} \otimes P_a^{(N-m)} +\frac{1}{N!} \sum_{\sigma\in S_N'} \mbox{\rm sgn}(\sigma) U_\sigma\:,\label{eq:vollanti}
\end{align}
where $S_N' \subset S_N$ are the permutations which do not decompose into separate permutations of $\{1,\ldots,m\}$ and $\{m+1,\ldots,N\}$, i.e.
\begin{equation}
\sigma\in S_N'\quad\text{if}\quad\sigma(\{1,2,\dots,m\})\cap\{m+1,\dots,N\}\neq\emptyset\:.
\end{equation}
To see that this is true, consider permutations $\sigma\in S_N\setminus S_N'$, which can be decomposed as $U_\sigma=U_\mu\otimes U_\nu$, where $\mu\in S_m$ and $\nu\in S_{N-m}$. Thus
\begin{equation}
\sum_{\sigma\in S_N\setminus S_N'} \mbox{\rm sgn}(\sigma) U_\sigma =\sum_{\mu\in S_m}\sum_{\nu\in S_{N-m}} \mbox{\rm sgn}(\mu) \,\mbox{\rm sgn}(\nu) \,U_\mu\otimes U_\nu= m!(N-m)!\, P_a^{(m)}\otimes P_a^{(N-m)}\:,
\end{equation}
which gives (\ref{eq:vollanti}). The same works for (\ref{eq:vollsymm}).

A bit of care is needed in constructing finitely supported Weyl sequences. To $\lambda \in \sigma(H_m^{(*)})$, let $\tilde{\psi}_n \in \ell^2_*(\Z^m)$ be a normalized Weyl sequence, i.e.\ $(H_m^{(*)}-\lambda) \tilde{\psi}_n \to 0$. Due to the boundedness of $H_m^{(*)}$ we can find sufficiently large cubes $\Lambda_n \subset \Z^m$ such that $(H_m^{(*)}-\lambda) \idty_{\Lambda_n} \tilde{\psi}_n \to 0$. Using that the $\Lambda_n$ are cubes we can see that $\psi_n := \idty_{\Lambda_n} \tilde{\psi}_n / \|\idty_{\Lambda_n} \tilde{\psi}_n\|$ is a normalized finitely supported Weyl sequence in $\ell^2_*(\Z^m)$ to $\lambda$ for $H_m^{(*)}$. Similarly we construct a normalized finitely supported Weyl sequence $\varphi_n$ to $\mu$ for $H_{N-m}^{(*)}$.

Similar to the proof of Proposition~\ref{prop1}, we define
\begin{equation} \label{eq:symmproduct}
\chi_n = P_*^{(N)} \left( \psi_n \otimes (T_{(c_n,\ldots,c_n)} \varphi_n) \right)
\end{equation}  
with $c_n$ sufficiently large to guarantee a wide separation of the supports of $\psi_n$ and $T_{(c_n,\ldots,c_n)} \varphi_n$. But note that the tensor product of the symmetric (or antisymmetric) states does not have the full symmetry in $N$ variables, so that in (\ref{eq:symmproduct}) we need to project into the symmetry subspace.

As a result $\chi_n$ is not normalized. However, if we can show that
\begin{equation} \label{eq:weyllowerbound}
\| \chi_n \| \ge c > 0
\end{equation}
uniformly in $n$, then the arguments in the proof of Propostion~\ref{prop1} would go through as before (as the $\chi_n$ could be renormalized). Thus it remains to show (\ref{eq:weyllowerbound}).

To simplify notation let us set
\begin{equation*}
K_n :=\text{supp}\:\psi_n \quad \text{and}\quad L_n :=\text{supp}\:(T_{(c_n,\dots,c_n)} \varphi_n)\:.
\end{equation*}
Write $z=(x,y) \in \Z^N$ with $x\in \Z^m$ and $y\in \Z^{N-m}$ and consider the symmetric case (the same argument works for the antisymmetric case). Using (\ref{eq:vollsymm}) we find
\begin{eqnarray}
\|\chi_N\|^2 & = & \sum_{z \in\Z^N}|\chi_n(z)|^2\geq\sum_{x\in K_n,\;y\in L_n}|\chi_n(x,y)|^2 \\ 
& = & \sum_{K_n, \,L_n} |P_s^{(N)} \left( \psi_n \otimes (T_{(c_n,\ldots,c_n)} \varphi_n) \right)(x,y)|^2 \nonumber \\ 
& = & \sum_{K_n,\,L_n}\bigg|\bigg[\bigg(\frac{m!(N-m)!}{N!} P_s^{(m)} \otimes P_s^{(N-m)} + \frac{1}{N!} \sum_{\sigma\in S_N'} U_\sigma\bigg) \nonumber \\ 
& & \qquad \qquad \big(\psi_n \otimes (T_{(c_n,\ldots,c_n)} \varphi_n)\big)\bigg](x,y)\bigg|^2\:. \nonumber
\end{eqnarray}  

It is clear that 
\begin{equation} 
(P_s^{(m)} \otimes P_s^{(N-m)}) (\psi_n \otimes (T_{(c_n,\ldots,c_n)} \varphi_n)) = \psi_n \otimes (T_{(c_n,\ldots,c_n)} \varphi_n).
\end{equation}
Moreover, we only sum over the supports of the Weyl sequences, which are disjoint by application of the shift operator. Thus it is easy to see that for any $\sigma\in S_N'$,
\begin{equation}
U_\sigma (\psi_n \otimes (T_{(c_n,\ldots,c_n)} \varphi_n))(x,y) =0\quad \mbox{for all $x \in K_n, \;y \in L_n$}\:.
\end{equation}

So, we conclude that
\begin{equation}
\|\chi_n\|\geq\frac{m!(N-m)!}{N!}>0
\end{equation}
independent of $n$.
\end{proof}

By repeated application this extends to decomposition into arbitrarily many clusters:

\begin{cor}
Let $H_N$ be defined as above. Then for any partition of $N$, i.e.\ positive integers $(n_1,n_2,\dots,n_k)$ such that $\sum_{i=1}^k n_i=N$ we have
\begin{equation}
\sigma(H_{n_1}^{(*)}) + \sigma(H_{n_2}^{(*)}) + \ldots + \sigma(H_{n_k}^{(*)}) \subset \sigma(H_N^{(*)})\;,
\end{equation}
where $* \in \{s,a\}$.
\label{corocluster}
\end{cor}

\section{The structure of the spectrum} \label{sec:gap}

In this section we return to the study of the structure of the spectra of the operators 
\begin{equation} \label{eq:HZNOperators}
H_{\Z}^N=-\frac{1}{2\Delta}h_0^{(\mathcal{X}^N)}+W+N\cdot\idty \quad \mbox{on $\ell^2(\mathcal{X}^N)$}\,.
\end{equation}
This was initiated in Section~\ref{sec:droplet} by identifying the droplet spectrum $\delta_N \subset \sigma(H_{\Z}^N)$, which in the sense of scattering theory corresponds to the channel in which all $N$ particles form a single cluster and can be exactly calculated as (\ref{eq:dropletband}) using the Bethe ansatz. We can now combine this with the general results of the previous section to find additional subsets of the spectrum corresponding to scattering channels for other cluster decompositions of the particles. Here the knowledge of the exact form of the droplet bands (\ref{eq:dropletband}), in particular the monotonicity and convexity properties of its endpoints as functions of $N$, will allow for a very explicit and simple characterization of the arising {\it cluster bands}, see Section~\ref{subsec:cluster}.

We will not attempt to prove for general $N$ that the union of cluster bands gives the entire spectrum of $H_{\Z}^N$. In particular, from the existence of gaps between the cluster bands we can not conclude that $\sigma(H_{\Z}^N)$ has gaps. Nevertheless, In Section~\ref{subsec:gaps} we will use other methods to prove the existence of gaps in the spectra of $H_{\Z}^N$ as well as in the spectrum of $H_{\Z}$, requiring that $\Delta$ is sufficiently large.

\subsection{Cluster bands} \label{subsec:cluster}

In Section~\ref{subsec:HVZ} choose $h = -\frac{1}{2\Delta} h_{0,\Z}$ from (\ref{eq:discLap}) for the single particle operator and the next-neighbor interaction potential $Q$ from (\ref{eq:potential}). Then, as established in Section~\ref{sec:Schrodinger}, $H_{\Z}^N$ is unitarily equivalent to the restriction of the $N$-body Hamiltonian from (\ref{eq:NpartHam}) to the antisymmetric subspace $\ell^2_a(\Z^N)$. Thus Proposition~\ref{prop:symm} and Corollary~\ref{corocluster} apply. In particular, using that $\delta_{n_i} \subset \sigma(H_{\Z}^{n_i})$ by (\ref{eq:dropletband}), we conclude that 
\begin{equation}
\delta_{(n_1,\ldots,n_k)} := \delta_{n_1}+\delta_{n_2}+\dots+\delta_{n_k}\subset\sigma( H_{\Z}^N)
\end{equation}
for every partition $(n_1, n_2, \ldots, n_k)$ of $N$.

Each $\delta_{(n_1,\ldots,n_k)}$ is an interval. As the number of possible partitions increases rapidly in $N$, this might lead to the idea that the {\it total cluster spectrum}
\begin{equation} \label{eq:barbecue}
\mathcal{C}(N) := \bigcup_{k=1}^N \bigcup_{\begin{array}{cc} (n_1,\ldots,n_k) \in \N^k \\ n_1+\ldots+n_k =N \end{array}} \delta_{(n_1,\ldots,n_k)} \subset \sigma(H_{\Z}^N)
\end{equation}
might become very complicated for large $N$. However, we will show next that $\mathcal{C}(N)$ is a union of $N$ bands, each corresponding to partitions with a given number $k \in \{1,\ldots,N\}$ of components. To this end, let
\begin{equation}
\mathfrak{P}_{k}(N):=\left\{(n_1,n_2,\dots,n_k)\in\N^k:\sum_{i=1}^k n_i=N\right\}
\end{equation}
be the set of all partitions of $N$ with exactly $k$ components. Define the  {\it $k$-cluster band} $\mathcal{C}_k(N)$ as
\begin{equation}
\mathcal{C}_{k}(N):=\bigcup_{(n_1,\dots,n_k)\in\mathfrak{P}_k(N)}\delta_{(n_1,n_2,\dots,n_k)}\:.
\end{equation}

The following proposition will justify the name  ``cluster {\em{band}}\," and also provide a simple formula for $\mathcal{C}_k(N)$. In particular, this verifies that the total cluster spectrum $\mathcal{C}(N) = \bigcup_{k=1}^N \mathcal{C}_k(N)$ is a union of at most $N$ bands.

\begin{prop} \label{prop:clusterband}
Let $\mathcal{C}_{k}(N)$ be the $k$-cluster band for $N$ particles. Then $\mathcal{C}_{k}(N)$ is an interval given by
\begin{equation}
\mathcal{C}_k(N)=\underbrace{\delta_1+\delta_1+\dots+\delta_1}_{k-1\:\text{times}}+\delta_{N-k+1}=\left[(k-1)(1-\frac{1}{\Delta}), (k-1)(1+\frac{1}{\Delta})\right]+\delta_{N-k+1}\:.
\label{eq:clusterband}
\end{equation}
\end{prop}

\begin{proof}
First, we show that $\mathcal{C}_k(N)$ is an interval. To this end, observe that the lower/upper boundary of $\delta_N$ as given by (\ref{eq:dropletband}) is increasing/decreasing in $N$,
\begin{equation}
\frac{\text{d}}{\text{d}N}\left(\tanh(\rho)\cdot\frac{\cosh(N\rho)\mp1}{\sinh(N\rho)}\right)=\begin{cases}&\frac{\rho}{2}\frac{\tanh(\rho)}{\cosh^2(N\rho/2)}>0\\-&\frac{\rho}{2}\frac{\tanh(\rho)}{\sinh^2(N\rho/2)}<0\end{cases}\:,
\end{equation}
which shows that all $\delta_N$ have the point
\begin{equation}
\delta_\infty:=\bigcap_{N=1}^\infty\delta_N=\left\{\sqrt{1-\frac{1}{\Delta^2}}\right\}
\end{equation}
in common. Thus, for any $\delta_{(n_1,\dots,n_k)}$ we know that
\begin{equation}
k\cdot\sqrt{1-\frac{1}{\Delta^2}}\in\delta_{(n_1,\dots,n_k)}\:.
\end{equation}
But a union of intervals, which all have at least one point in common has to be an interval itself. 

It remains to prove (\ref{eq:clusterband}), which will be done by explicitly determining the minimum and maximum of $\mathcal{C}_k(N)$. The cases $k=1$ and $k=N$ are trivial.

At first, we argue that for any $\ell \in\{1,\dots,N-1\}$ it holds that
\begin{equation}
\min\left(\delta_1+\delta_{N-1}\right)\leq\min\left(\delta_{\ell}+\delta_{N-\ell}\right)\quad\text{and}\quad\max\left(\delta_1+\delta_{N-1}\right)\geq\max\left(\delta_{\ell}+\delta_{N-\ell}\right)\:,
\label{eq:2cluster}
\end{equation}
which corresponds to the case of $k=2$ clusters.
To see this, just observe that the function for the lower boundary of the droplet bands is concave, respectively convex for the upper boundary,
\begin{equation}
\frac{\text{d}^2}{\text{d}N^2}\left(\tanh(\rho)\frac{\cosh(N\rho)\mp1}{\sinh(N\rho)}\right)=\mp\frac{\rho^2\tanh(\rho)\left(1\mp\cosh(N\rho)\right)^2}{\sinh^3(N\rho)}\begin{cases}<0\\>0\end{cases}\:.
\end{equation}
Proceeding inductively, assume that $\eqref{eq:clusterband}$ holds for up to $k$ clusters. Then for any combination of $k+1$ droplet bands with $\sum_{i=1}^{k+1}n_i=N$, we see that
\begin{align}
&\min\delta_{n_1}+\min\delta_{n_2}+\dots+\min\delta_{n_k} +\min\delta_{n_{k+1}}\\ 
\nonumber \geq&\underbrace{\min\delta_1+\min\delta_1+\dots+\min\delta_1}_{k-1\:\text{times}}+\min\delta_{N-n_{k+1}-(k-1)}+\min\delta_{n_{k+1}}\\
\nonumber \geq&\underbrace{\min\delta_1+\min\delta_1+\dots+\min\delta_1}_{k-1\:\text{times}}+\min\delta_1+\min\delta_{N-k}\:,
\end{align}
where (\ref{eq:2cluster}) was used on the last two terms. This is the desired result for the minimum.  The argument for the maximum works analogously. 
\end{proof}

The identities (\ref{eq:dropletband}) and (\ref{eq:clusterband}) give explicit expressions for the droplet bands $\delta_N$ and total cluster spectra $\mathcal{C}(N)$ in terms of $N$ and $\rho$. We find it instructive and useful for later reference to also provide these formulas in terms of the anisotropy parameter $\Delta = \cosh \rho$, at least for particle numbers up to $N=3$. In addition to the formulas (\ref{eq:delta1}) and (\ref{eq:delta2}) for $\delta_1$ and $\delta_2$ we find from hyperbolic identities that
\begin{equation} \label{eq:delta3}
\delta_3 = \left[ 1- \frac{1}{2\Delta^2-\Delta}, 1- \frac{1}{2\Delta^2+\Delta} \right],
\end{equation}
which yields the cluster spectra
\begin{eqnarray}
\mathcal{C}(1) & = & \left[1-\frac{1}{\Delta}, 1+\frac{1}{\Delta} \right], \\
\mathcal{C}(2) & = & \left[ 1-\frac{1}{\Delta^2}, 1 \right] \cup \left[ 2-\frac{2}{\Delta}, 2+\frac{2}{\Delta} \right], \label{eq:C2} \\
\mathcal{C}(3) & = & \left[ 1- \frac{1}{2\Delta^2-\Delta}, 1- \frac{1}{2\Delta^2+\Delta} \right] \cup \left[ 2-\frac{2}{\Delta} - \frac{1}{\Delta^2}, 2 + \frac{2}{\Delta} \right] \\ & & \mbox{} \nonumber \cup \left[ 3 - \frac{3}{\Delta}, 3 + \frac{3}{\Delta} \right].
\end{eqnarray}
Note, in particular, that for $N=2$ and $N=3$ the droplet band is not separated from the rest of the cluster spectrum, unless $\Delta$ is sufficiently large.

We can also use the above results to identify a large subset of the spectrum of the Hamiltonian $H_{\Z}$ of the infinite XXZ chain, which by (\ref{eq:infvolXXZ}) is the direct sum of the $H_{\Z}^N$.

\begin{cor} \label{cor:XXZspec}
For every $\Delta>1$ it holds that
\begin{equation} \label{eq:shrimpandgrits}
\bigcup_{N=0}^{\infty} \mathcal{C}(N) = \bigcup_{k=0}^{\infty} \left[ k(1-\frac{1}{\Delta}), k(1+\frac{1}{\Delta}) \right] \subset \sigma(H_{\Z}),
\end{equation}
where we have set $\mathcal{C}(0):=\{0\}$. 
\end{cor}

Note that we are not claiming that $\mathcal{C}(N)$ coincides with $\left[ N(1-\frac{1}{\Delta}), N(1+\frac{1}{\Delta}) \right]$ and that this is indeed not true.

\begin{proof}
As $\mathcal{C}(N) \subset \sigma(H_{\Z}^N)$ we have by (\ref{eq:infvolXXZ}) that $\bigcup_{N\ge 0} \mathcal{C}(N) \subset \sigma(H_{\Z})$, where $N=0$ corresponds to the vacuum energy $E=0$. It remains to show equality of the two unions in (\ref{eq:shrimpandgrits}), where we only need to consider $N\ge 1$ and $k\ge 1$. One of the necessary inclusions follows from
\begin{equation}
\left[ k(1-\frac{1}{\Delta}), k (1+\frac{1}{\Delta}) \right] = \underbrace{\delta_1+\dots+\delta_1}_{k\:\text{times}} = \mathcal{C}_k(k) \subset \mathcal{C}(k).
\end{equation}
On the other hand, as $\delta_k \subset \delta_1$ for all $k$,
\begin{eqnarray}
\mathcal{C}(N) = \bigcup_{k=1}^N\mathcal{C}_k(N) & = & \delta_N\cup(\delta_1+\delta_{N-1})\cup \dots \cup(\underbrace{\delta_1+\dots+\delta_1}_{N\:\text{times}}) \\
& \subset & \delta_1 \cup (\delta_1+\delta_1) \cup \dots\cup(\underbrace{\delta_1+\dots+\delta_1}_{N\:\text{times}}) \nonumber \\
& = & \bigcup_{k=1}^N \left[ k(1-\frac{1}{\Delta}), k (1+\frac{1}{\Delta}) \right], \nonumber
\end{eqnarray}
which yields the second inclusion.
\end{proof}

It remains an open problem if equality holds in (\ref{eq:barbecue}), i.e.\ $\mathcal{C}(N) = \sigma(H_{\Z}^N)$ for all $N\in \N$, and thus (\ref{eq:shrimpandgrits}) gives the entire spectrum of the infinite XXZ chain. The only cases where we can verify this are the trivial case $N=1$ and the slightly less trivial case $N=2$. In the latter one can use tools from the theory of Jacobi matrices to determine the exact spectrum of the fiber operators $\widehat{H}^2(\vartheta)$ from (\ref{eq:BlochN2}), namely
\begin{equation}
\sigma(\widehat{H}^2(\vartheta)) = \left\{ 1- \frac{1+\cos \vartheta}{2\Delta^2} \right\} \cup \left[ 2 \left( 1- \sqrt{\frac{1+\cos \vartheta}{2\Delta^2}} \right), 2 \left( 1+ \sqrt{\frac{1+\cos \vartheta}{2\Delta^2}} \right) \right] \end{equation}
for all $\vartheta \in [-\pi, \pi)$, whose union $\sigma(H_{\Z}^2)$ is indeed $\mathcal{C}(2)$ as given by (\ref{eq:C2}).

The validity of $\mathcal{C}(N) = \sigma(H_{\Z}^N)$ for all $N$ should certainly be a consequence of completeness of the Bethe ansatz for the XXZ chain. The latter is understood to mean much more, namely an explicit diagonalization of the XXZ chain Hamiltonian in terms of a Plancherel formula, and therefore the identification of all generalized eigenfunctions and thus also the spectrum of $H_{\Z}$ and its $N$-particle restrictions. For work on this, which we haven't been able to completely verify, we refer to the papers mentioned in the introduction.

\subsection{Separation of energy levels} \label{subsec:gaps}

In this section we will prove the existence of spectral gaps for $H_{\Z}^N$ and $H_{\Z}$. Such results would follow from completeness of the Bethe ansatz and the resulting formulas $\mathcal{C}(N) = \sigma(H_{\Z}^N)$, $N=1,2,\ldots$, and thus $\bigcup_N \mathcal{C}(N) = \sigma(H_{\Z})$, given our explicit knowledge of the cluster bands. To keep our arguments ``soft'' we will instead provide more direct arguments for the existence of spectral gaps.

We start with a simple perturbation theoretic observation, where the kinetic energy term $-\frac{1}{2\Delta}h_0^{(\mathcal{X}^N)}$ is considered as a perturbation of the potential $W+N\cdot\idty$ in (\ref{eq:HZNOperators}). The latter only takes the values $1,\ldots,N$ and thus
\begin{equation} 
\sigma(W+N\cdot \idty) = \{1,\ldots,N\}.
\end{equation}

On the other hand it is easy to see that $\|h_0^{(\mathcal{X}^N)}\| \le 2N$, for example by using the representation (\ref{eq:graphlaplace}) of the unitarily equivalent operator $h_0^{(\Z\times \N^m, \Gamma)}$ and that $\|T_j\|= \|T_j^*\|=1$. In fact, using a Weyl sequence argument it is not much harder to see that $\|h_0^{(\mathcal{X}^N)}\| = 2N$. 

Thus we can apply the general fact that for two self-adjoint operators $A$, $B$ it holds that $\sigma(A+B) \subset \sigma(A) + [-\|B\|, \|B\|]$ to conclude that
\begin{equation} \label{eq:simplegap}
\sigma(H_{\Z}^N) \subset \bigcup_{k=1}^{N} \left[k-\frac{N}{\Delta},k+\frac{N}{\Delta}\right]\:.
\end{equation}

Let us compare this with the lower bound $\mathcal{C}(N) = \bigcup_{k=1}^N \mathcal{C}_k(N) \subset \sigma(H_{\Z}^N)$ which we found above. It is indeed easy to see that 
\begin{equation} 
\mathcal{C}_k(N) \subset \left[k-\frac{N} {\Delta},k+\frac{N}{\Delta}\right] 
\end{equation}
for all $k$ and $N$ (use (\ref{eq:clusterband}) and $\delta_{N-k+1}\subset\delta_1$), which is an equality for the highest energy band $\mathcal{C}_N(N)$. 

Thus we find that for any fixed $N$ and $\Delta >2N$ the spectrum $\sigma(H_{\Z}^N)$ decomposes into $N$ non-trivial parts $\sigma_k$, $1\le k \le N$, separated by gaps such that
\begin{equation}
\mathcal{C}_k(N) \subset \sigma_k \subset \left[ k- \frac{N}{\Delta}, k + \frac{N}{\Delta} \right].
\end{equation}
We see that $\sigma_N = [N-N/\Delta, N+N/\Delta]$, but generally can not identify $\sigma_N$ (or even argue that it is an 
interval).

Also, (\ref{eq:simplegap}) does not establish the existence of gaps in $\sigma(H_{\Z})$, independent of how large $\Delta$ is chosen, as the union of the right hand sides over $N$ covers all of $\R$. All we can say about gaps of $H_{\Z}$ so far is that for every $\Delta>1$ is has a gap of size $1-1/\Delta$ above its ground state energy $E_0=0$: By (\ref{eq:specminN}) we have $\min \sigma(H_{\Z}^N) = \min \delta_N$, which is increasing in $N$ as shown in the proof of Proposition~\ref{prop:clusterband}. Thus
\begin{equation}
\min \left( \sigma(H_{\Z}) \setminus \{0\} \right) = \min \delta_1 = 1- \frac{1}{\Delta}.
\end{equation}

Our final goal of this section is to show that $H_{\Z}$ has an additional gap above the union of the droplet spectra if $\Delta$ is sufficiently large. In the proof of this we will have to go beyond the simple norm estimates which led to (\ref{eq:simplegap}).

\begin{prop} \label{prop:existencegap}
For $\Delta>3$ the spectrum $\sigma(H_\Z)$ has a gap of width $1-\frac{3}{\Delta}$ above the first band, i.e.\  
\begin{equation}
\sigma(H_{\Z}) \cap \left(-\infty, 2-\frac{2}{\Delta} \right) = \{0\} \cup \left[ 1- \frac{1}{\Delta}, 1 + \frac{1}{\Delta} \right].
\end{equation}
\end{prop}

\begin{proof}
We know that $\sigma(H_{\Z}) = \overline{\cup_N \sigma(H_{\Z}^N)}$, $\sigma(H_{\Z}^0) = \{0\}$, $\sigma(H_{\Z}^1) = [1-1/\Delta, 1+1/\Delta]$ and $\delta_N \subset [1-1/\Delta, 1+1/\Delta]$ for all $N$. This reduces the proof of the theorem to showing that
\begin{equation} \label{eq:toshow1}
\sigma(H_{\Z}^N) \setminus \delta_N \subset \left[ 2 - \frac{2}{\Delta}, \infty \right) \quad \mbox{for all $N\ge 2$}.
\end{equation}
As $H_{\Z}^N$ is unitarily equivalent to the direct integral of $\widehat{H}^N(\vartheta)$ over $\vartheta \in [-\pi,\pi)$ and $\delta_N$ is the union of the $E_N(\vartheta)$ given by (\ref{eq:dropleteigval}), the proof of (\ref{eq:toshow1}) reduces to
\begin{equation} \label{eq:toshow2}
\sigma(\widehat{H}^N(\vartheta)) \setminus \{ E_N(\vartheta) \} \subset \left[ 2 - \frac{2}{\Delta}, \infty \right) \quad \mbox{for all $\vartheta \in [-\pi, \pi)$, $N\ge 2$.}
\end{equation}
Recall here that $\widehat{H}^N(\vartheta)$ acts in $\ell^2(\N^m)$, $m=N-1$, via (\ref{eq:fiberoperator}). We will prove (\ref{eq:toshow2}) by a rank-one perturbation argument. For this let $A = \langle e_{(1,1,\ldots,1)}, \cdot \rangle e_{(1,1,\ldots,1)}$ be the orthogonal projection onto the span of the canonical basis vector $e_{(1,1,\ldots,1)}$ in $\ell^2(\N^m)$. We will show that
\begin{equation} \label{eq:toshow3}
\widehat{H}^N(\vartheta) + A \ge 2 - \frac{2}{\Delta}
\end{equation}
for all $N$ and $\vartheta$. That $A$ is rank-one then implies that $\widehat{H}^N(\vartheta)$ has at most one simple eigenvalue below $2-2/\Delta$, which by Proposition~\ref{prop:eigval} is $E_N(\vartheta)$ (the latter lies in $\delta_N \subset [1-1/\Delta,1+1/\Delta]$ and thus indeed below $2-2/\Delta$ due to $\Delta >3$).

The quadratic form of $\widehat{H}_N(\vartheta)+A$ is given by
\begin{align} \label{eq:perturbed}
\langle u,(\widehat{H}_N&(\vartheta)+A)u\rangle = \\ & -\frac{1}{2\Delta}\left(\langle u,(e^{i\vartheta}T_1^*+e^{-i\vartheta}T_1)u\rangle+\sum_{j=1}^{m-1}\langle u,(T_{j+1}^*T_j+T_j^*T_{j+1})u\rangle +\langle u,(T_m^*+T_m)u\rangle\right) \notag \\ & +\sum_{j=1}^m\langle u,P_j u\rangle+\langle u,u\rangle+\langle u,Au\rangle\:. \notag
\end{align}
Observe that 
\begin{align}
\langle u,(e^{i\vartheta}T_1^*+e^{-i\vartheta}T_1)u\rangle&=2 \,\mbox{\rm Re}\,\langle u, e^{i\vartheta}T_1^* u\rangle\leq\langle u,u\rangle+\langle T_1^* u,T_1^*u\rangle\\
\langle u,(T^*_{j+1}T_j+T^*_j T_{j+1})u\rangle&=2 \,\mbox{\rm Re} \,\langle T_{j+1}^* u,T_j^* u\rangle\leq \langle T_{j}^* u,T_j^* u\rangle+\langle T_{j+1}^* u,T_{j+1}^* u\rangle\\
\langle u,(T_m^*+T_m)u\rangle&=2 \,\mbox{\rm Re} \,\langle u, T_m^*u\rangle\leq\langle u,u\rangle+\langle T_m^* u,T_m^* u\rangle\:,
\end{align}
where we have used in the second line that $T_{j+1}^*$ and $T_j$ (as well as $T_{j+1}$ and $T_j^*$) commute. Now, we use these inequalities to estimate the kinetic term in \eqref{eq:perturbed} from below and, after some simplification, get
\begin{equation} \label{eq:estimate}
\langle u,(\widehat{H}_N(\vartheta)+A)u\rangle\geq -\frac{1}{\Delta}\left(\langle u,u\rangle+\sum_{j=1}^m\langle T_j^* u,T_j^* u\rangle\right)+\sum_{j=1}^m\langle u,P_j u\rangle+\langle u,u\rangle+\langle u,Au\rangle\:.
\end{equation}
The operator $T_j^*$ truncates after a left shift in the $j$-th coordinate, so that $\langle T_j^* u,T_j^* u\rangle = \langle u,P_j u\rangle$. Therefore
\begin{equation} \label{eq:gutenacht}
\langle u,(\widehat{H}_N(\vartheta)+A)u\rangle\geq\left(1-\frac{1}{\Delta}\right)\langle u,u\rangle+\left(1-\frac{1}{\Delta}\right)\sum_{j=1}^m\langle u,P_j u\rangle+\langle u,A u\rangle\:.
\end{equation}
The quadratic form on the right hand side is diagonal, i.e.\ of the form $\sum_{n\in \N^m} \alpha_n |u_n|^2$. Thus for the proof of (\ref{eq:toshow3}) it remains to be shown that $\alpha_n \ge 2 - 2/\Delta$ for all $n\in \N^m$.

For $n=(1,1,\ldots,1)$ we get contributions from the first and last terms on the right hand side of (\ref{eq:gutenacht}), showing that $\alpha_{(1,1,\ldots,1)} \ge 2-1/\Delta$. For all other $n$ we have $n_j \ge 2$ for at least one $j \in \{1,\ldots,m\}$. Thus $\langle u, P_j u \rangle \ge |u_n|^2$, so that contributions from the first and second term in (\ref{eq:gutenacht}) yield $\alpha_n \ge 2-2/\Delta$. 
\end{proof} 

We conclude this section with some remarks:

(i) We already know that $\delta_1\cup(\delta_1+\delta_1) = [1-1/\Delta, 1+1/\Delta] \cup [2-2/\Delta, 2+2/\Delta] \subset\sigma(H_\Z)$. For $\Delta \le 3$ this is a single interval without gap. Thus the condition $\Delta>3$ for the existence of a gap above the first band of ${\sigma}(H_\Z)$ is sharp.

(ii) Let us ignore the somewhat trivial case $N=1$, for which the spectrum of $H_{\Z}^1$ consists of the single band $\delta_1 = [1-1/\Delta, 1+1/\Delta]$. For $N\ge 2$ the above argument shows that
\begin{equation}
\sigma(H_{\Z}^N) \cap \left(-\infty, 2-\frac{2}{\Delta} \right) = \{0\} \cup \delta_N \subset \{0\} \cup \left[ 1-\frac{1}{\Delta^2}, 1 \right],
\end{equation}
so that $(1,2-2/\Delta)$ is a spectral gap of $H_{\Z}^N$ uniformly for all $N\ge 2$. For this to be non-trivial we only need $\Delta>2$.

(iii) We believe that for {\it any} $\Delta>1$ and asymptotically in large particle number the $H_{\Z}^N$ will have a uniform spectral gap above the droplet band, more precisely: For all $\Delta>1$ there exist $N_0$ and a non-trivial open interval $I_0$ such that, for all $N\ge N_0$, $I_0$ lies in a spectral gap above the droplet band $\delta_N$ of $H_{\Z}^N$. With the above reasoning, using that $\delta_N \to \{\sqrt{1-1/\Delta^2}\}$ this can be verified for $\Delta>5/3$. For smaller $\Delta$ this would require more refined quadratic form bounds in the above proof. This would be an infinite volume version of the main result in \cite{NS}, where it is shown for the finite XXZ chain with $\Delta>1$ and suitable boundary conditions that the droplet spectrum is uniformly separated from higher excitations in the limit of a large number of down-spins (particles).

\section{Concluding Remarks} \label{sec:remarks}

As mentioned in the Introduction, our motivation for looking at the XXZ spin chain stems from our believe that it provides a promising model for the investigation of many-body localization phenomena. While this will be left to future work, we will use this concluding section to elaborate some more on our reasons for this belief.

Specifically, one might want to consider the infinite XXZ chain $H_{\Z}(\nu)$ in exterior field given by (\ref{eq:infvolXXZmag}), where $\nu = (\nu_x)_{x\in \Z}$ are i.i.d.\ random variables with common distribution $\mu$ such that
\begin{equation} \label{eq:suppmu}
\mbox{supp}\, \mu = [0, \nu_{max}].
\end{equation}
$H_{\Z}(\nu)$ is unitarily equivalent to the direct sum of the operators $H_{\Z}^N(\nu)$ given by (\ref{eq:equivform}). As explained in Section~\ref{sec:Schrodinger}, the latter are Hamiltonians of interacting $N$-particle Anderson models with attractive next-neighbor interaction $Q$ given by (\ref{eq:potential}) and restricted to the fermionic subspace. 

A first observation is that each $H_{\Z}^N(\nu)$ is ergodic with respect to the translation $T$ on $\ell^2({\mathcal X}^N)$ given by $Te_{(x_1,\ldots,x_N)} = e_{(x_1+1,\ldots,x_N+1)}$ and the shift $(S\nu)_x = {\nu}_{x+1}$ of the random parameters, that is
\begin{equation}
T^{-1} H_{\Z}^N(\nu) T = H_{\Z}^N(S\nu).
\end{equation}
Thus $H_{\Z}^N(\nu)$ has deterministic spectrum, i.e.\ there exists $\Sigma^N \subset \R$ such that
\begin{equation}
\sigma(H_{\Z}^N(\nu)) = \Sigma^N \quad \mbox{for a.e.\ $\nu$},
\end{equation}
see Chapter~V.2 of \cite{CarmonaLacroix}.

One can actually say more and, under the assumption (\ref{eq:suppmu}), explicitly describe the almost sure spectrum of $H_{\Z}^N(\nu)$ as
\begin{equation} \label{eq:asspec}
\Sigma^N = \sigma(H_{\Z}^N) + [0,N \nu_{max}].
\end{equation}
A detailed proof of this is given in \cite{F} and follows arguments well known from the study of the Anderson model, e.g.\ Proposition~V.3.5 of \cite{CarmonaLacroix} or Theorem~2 \cite{Stolz}.

As a consequence, the spectrum of the infinite random XXZ chain is almost surely given by
\begin{eqnarray}
\sigma(H_{\Z}(\nu)) & = & \{0\} \cup \bigcup_{N=1}^{\infty} \Sigma^N \\
& = & \{0\} \cup \bigcup_{N=1}^{\infty} \left( \sigma(H_{\Z}^N) + [0,N\nu_{max}] \right) \nonumber \\
& = & \{0\} \cup \left[ 1-\frac{1}{\Delta}, \infty \right), \nonumber
\end{eqnarray}
where the latter uses $\delta_N \subset \sigma(H_{\Z}^N)$, the nesting property of the droplet bands $\delta_N$ and $\min \delta_1 = 1-1/\Delta$. Thus, for any $\nu_{\max}>0$, the randomness smears out all the spectral gaps of $H_{\Z}(\nu)$ other than the ground state gap.

However, for fixed $N\ge 2$, $\nu_{max}>0$ sufficiently small and $\Delta$ sufficiently large, (\ref{eq:asspec}) and the inclusion (\ref{eq:toshow1}) shown in the proof of Proposition~\ref{prop:existencegap} implies that $\sigma(H_{\Z}^N(\nu))$ has an additional spectral gap above the {\it extended droplet band} $\delta_N + [0, N\nu_{max}]$.

Regarding localization properties, a first question to ask is if for fixed $N$ the spectrum of $H_{\Z}^N(\nu)$ is localized near the almost sure spectral minimum, which, by (\ref{eq:asspec}), is given by the minimum of the droplet band. This is a question about localization for the interacting fermionic $N$-body Anderson model, but the localization regime considered here is not covered by known results on the $N$-body Anderson model, e.g.\ \cite{AizenmanWarzel, ChulaevskySuhov}. Nevertheless, we believe that the explicit knowledge we have about the droplet spectrum and its generalized eigenfunctions and the fact that the low lying spectrum of $H_{\Z}^N(\nu)$ can be considered as one-dimensional {\it surface spectrum} (compare with Figure~\ref{fig:1} for the case $N=2$) provide tools which will allow to tackle this problem.

Moreover, one could hope that the generalized eigenfunctions for the droplet bands obtained via the Bethe ansatz provide enough quantitative information to control the $N$-dependence in the localization proof for $H_{\Z}^N(\nu)$. One goal could be to show that the infinite random XXZ chain $H_{\Z}(\nu)$ has a region of pure point spectrum which extends above $\sqrt{1-\Delta^{-2}}$, the limiting point of the droplet bands as $N \to \infty$. This would mean the existence of localized states of arbitrarily large particle number (or droplet size), which can be considered as a form of many-body localization.



\end{document}